\documentclass[letterpaper, 10 pt, conference]{ieeeconf}

\usepackage[utf8]{inputenc}
\usepackage{amsmath,amssymb,amsfonts,mathtools}

\usepackage{amsthm}
\usepackage{xcolor}
\usepackage{graphicx}
\usepackage{float}
\usepackage{babel}
\usepackage{algorithm}
\usepackage{algorithmic}
\usepackage{booktabs}
\usepackage{caption}
\usepackage{subcaption}
\usepackage[hidelinks]{hyperref}

\newtheorem{theorem}{Theorem}
\newtheorem{lemma}{Lemma}
\newtheorem{corollary}{Corollary}
\newtheorem{proposition}{Proposition}
\newtheorem{definition}{Definition}
\newtheorem{remark}{Remark}
\newtheorem{assumption}{Assumption}

\newcommand{\cX}{\mathcal{X}}
\newcommand{\cU}{\mathcal{U}}
\newcommand{\cY}{\mathcal{Y}}
\newcommand{\cB}{\mathcal{B}}
\newcommand{\cK}{\mathcal{K}}
\newcommand{\cO}{\mathcal{O}}
\newcommand{\cH}{\mathcal{H}}
\newcommand{\cN}{\mathcal{N}}
\newcommand{\cS}{\mathcal{S}}
\newcommand{\cA}{\mathcal{A}}
\newcommand{\cZ}{\mathcal{Z}}
\newcommand{\DeltaS}{\Delta(\cS)}
\newcommand{\diag}{\operatorname{diag}}

\newcommand{\R}{\mathbb{R}}
\newcommand{\E}{\mathbb{E}}
\newcommand{\TV}{\mathrm{TV}}

\newcommand{\dW}{d_{W_1}}

\newcommand{\sigy}{\sigma_y}
\newcommand{\sigu}{\sigma_u}
\newcommand{\norm}[1]{\left\lVert #1 \right\rVert}
\newcommand{\diam}{\operatorname{diam}}
\newcommand{\supp}{\operatorname{supp}}

\title{Finite Reliability Representations: Noise-Calibrated Belief-Space
Covers for Reliable Decision-Making}

\author{
Hyung-Jin Yoon$^{\dagger}$ and
Hunmin Kim$^{\ddagger}$
\thanks{$^{\dagger}$H.-J. Yoon is with the
Department of Mechanical and Nuclear Engineering, Tennessee Technological
University, Cookeville, TN, USA.}%
\thanks{$^{\ddagger}$H. Kim is with the School of Engineering, Department
of Electrical and Computer Engineering, Mercer University, Macon, GA, USA.}%
\thanks{This work was supported by internal funding at Tennessee
Technological University.}
}

\begin{document}
\maketitle

\begin{abstract}
Physical sensing and actuation noise floors should inform how much belief
resolution a decision-making system can reliably use. We introduce Finite
Reliability Representations (FRR), a framework for covering belief spaces by
reliability cells: regions within which the optimal action-value function
\(Q^*(b,u)\) varies by at most a tolerance \(\varepsilon\), uniformly over
actions. The framework is formulated on beliefs rather than states and uses a
cover rather than an equivalence quotient, because approximate
decision-closeness is not transitive in general.
A central technical point is that noisy Bayesian updates should not be
treated as globally contractive on arbitrary beliefs. We therefore separate
three objects: the fixed-observation filter map, the predictive observation
law, and the controlled belief-transition kernel. For nonlinear
continuous-state systems, FRR is obtained under a reachable-set Lipschitz
modulus for the belief-transition kernel. For finite-state POMDPs, the same
construction becomes exact on the belief simplex: prediction is linear,
Bayesian correction is a normalized positive linear map, sensor noise enters
through observation-distribution distinguishability, and actuation
uncertainty enters through an action-execution channel.
Under the corresponding action-value Lipschitz condition, an FRR cover
supports a cell-constant policy whose suboptimality is bounded by
\(2\varepsilon/(1-\gamma)\). We also introduce reliability entropy, the
logarithm of the minimal number of reliability cells, as a measure of
certified decision-relevant belief complexity. The framework distinguishes
representation sufficiency from fundamental performance floors imposed by
sensing, process, and actuation noise. It applies to finite POMDPs,
linear-Gaussian filters, locally linearized nonlinear filters, and
particle-filter implementations through analytic or empirical certification
of reliability cells.
\end{abstract}

\section{Introduction}
\label{sec:intro}

Autonomous systems should not reason, plan, or control at a resolution
unsupported by the physical channels on which their closed-loop decisions
depend. Sensing noise limits what can be inferred from observations, while
actuation and process uncertainty limit what can be reliably executed. Yet
many methods for making partially observable decision problems tractable
introduce a representation tolerance chosen by the designer: a grid
resolution, a particle budget, a point-based approximation set, a
value-function error tolerance, or a behavioral similarity threshold. These
choices may be computationally reasonable, but they are not directly tied to
what the deployed sensing, process-noise, and action-execution channels can
actually support. A controller may therefore preserve distinctions that the
hardware cannot reliably use, or discard belief distinctions that remain
critical for safety and performance. The central question is not only how to
approximate a belief space, but what belief resolution is physically usable
for reliable decision-making.

Decision-making under partial observability is classically formalized as a
partially observable Markov decision process (POMDP), in which an agent
maintains a belief distribution over hidden states and selects actions to
maximize expected discounted reward
\cite{kaelbling1998planning,spaan2012pomdp}. Exact POMDP planning is
generally intractable for large or continuous state and observation spaces,
motivating belief discretization, value-function approximation, particle
representations, point-based value iteration, and reachability-guided
approximation methods
\cite{hauskrecht2000value,pineau2003point,kurniawati2008sarsop}. In
parallel, state aggregation and bisimulation metrics group states by
behavioral similarity, while information-state compression and
predictive-state representations reduce the dimensionality of the sufficient
statistic used for prediction or control
\cite{ferns2004metrics,ferns2011bisimulation,li2006unified,
littman2001predictive,singh2004predictive,roy2002exponential,
doshivelez2015bayesian}. These approaches provide powerful algorithmic
tools, but the representation criterion is typically algorithmic or
statistical rather than physical: the designer specifies what approximation
is acceptable.

This paper develops a theory of \emph{Finite Reliability Representations}
(FRR) for partially observed stochastic systems. The objective is to connect
belief-space representation to the physical noise floors of the deployed
system and to the decision sensitivity of the optimal action-value function.
The central object is a cover of the reachable belief space by reliability
cells. Within each cell, the optimal action-value function \(Q^*(b,u)\) is
allowed to vary by at most a prescribed decision tolerance \(\varepsilon\),
uniformly over actions. A policy may then choose actions using a
representative belief for each cell. If the decision diameter of every cell
is at most \(\varepsilon\), the resulting cell-constant policy is guaranteed
to be near-optimal relative to the optimal policy for the same noisy POMDP
model. Thus, the representation is finite not because the belief space is
artificially discretized, but because distinctions below the certified
decision tolerance are not needed for reliable action selection.

FRR differs from information-constrained decision-making. Information
bottleneck and bounded-rationality formulations typically constrain or
penalize the information a policy may extract from the state or observation,
thereby trading expected return against an information budget chosen by the
designer \cite{tishby1999information,genewein2015bounded}. FRR asks a
complementary question: how much belief information is physically supported
by the deployed sensing model, and how much of that information can affect
reliable action selection under the available action-execution model?
Throughout the paper, \(\sigma\) denotes a physical noise-floor parameter in
the planning model. In an observation-only setting, \(\sigma\) is the sensor
noise floor; in models with execution uncertainty, it may represent a
combined or vector-valued noise level such as sensing and actuation
uncertainty. The parameter \(\varepsilon\) is the decision tolerance used to
certify policy loss. Together, these parameters define a reliability level
\((\sigma,\varepsilon)\), rather than an arbitrary compression budget
detached from the hardware.

FRR is also distinct from entropy and information-gain criteria used in
active sensing, robotic exploration, and sensor placement
\cite{lindley1956measure,mackay1992information,krause2008near,
bourgault2002information,charrow2015information}. Entropy-based methods ask
where uncertainty can be reduced most. FRR asks whether the remaining
uncertainty can change any reliable action available to the deployed system.
A belief region may have high entropy but require no finer representation if
all beliefs in that region have nearly the same action values. Conversely,
a low-entropy region may require refinement if small belief changes cross a
decision-critical boundary. This distinction is why FRR uses decision
diameter, not posterior entropy alone, as the certification criterion.

A key technical point is that physical noise does not, in general, imply
global contraction of Bayesian filtering. The deterministic Bayesian update
\[
  \Phi_\sigma(b,u,y)
\]
maps a prior belief to a posterior belief after a particular observation
\(y\) has already occurred. This fixed-observation map need not be globally
contractive on the belief space. For example, if two beliefs have disjoint
or nearly singular support, reweighting by the same likelihood cannot
generally force their posteriors to contract. Therefore, the correct
smoothness object for dynamic programming is not the fixed-observation
Bayesian map alone, but the controlled belief-transition kernel
\[
  \cK_\sigma(\cdot\mid b,u),
\]
which describes the distribution of the next belief before the observation
is known and is the object that appears in the Bellman operator.

This distinction leads to the technical structure of the paper. First, we
define the predictive observation law
\[
  \cO_\sigma(\cdot\mid b,u),
\]
and use total variation distance between predictive observation laws to
quantify how distinguishable two beliefs are through the physical sensor
model. Sensor degradation is treated as a data-processing operation: adding
independent sensor noise cannot increase predictive observation
distinguishability. This provides a noise-calibrated notion of observational
resolution without requiring a global contraction property of the Bayesian
update.

Second, we incorporate actuation uncertainty through the controlled
transition model. In finite POMDPs this can be written explicitly as an
action-execution channel
\[
  G_{\sigma_u}(\tilde a\mid a),
\]
where \(a\) is the commanded action and \(\tilde a\) is the action actually
executed. This channel averages rewards and transitions over executed
actions. As execution uncertainty increases, distinct commanded actions may
become less distinguishable in their effects, which can flatten
action-value differences and permit coarser reliability cells. This
mechanism is different from sensor degradation: sensor noise limits what can
be inferred, while action-execution uncertainty limits what can be reliably
done.

Third, we impose a reachable-set regularity condition on the belief kernel.
Specifically, on the reachable belief set \(\cB_{\mathrm{reach}}\), we
assume that
\[
  W_1^{\cB}\!\left(
  \cK_\sigma(\cdot\mid b,u),
  \cK_\sigma(\cdot\mid b',u)
  \right)
  \le
  \beta_\sigma \dW(b,b')
\]
uniformly over admissible controls, where \(W_1^{\cB}\) is the
Wasserstein-1 distance over probability measures on belief space and
\(\dW\) is the Wasserstein-1 distance between beliefs. This condition is a
regularity assumption on the nominal planning model, not a claim that the
physical system is exactly known or that Bayesian filtering is globally
contractive. Under the discounted sensitivity condition
\[
  \gamma\beta_\sigma<1,
\]
we prove that the optimal value function and the optimal action-value
function are Lipschitz on the reachable belief set. This gives a direct
route from belief-kernel smoothness to bounded decision diameter.

The resulting FRR construction is simple. If a cover of
\(\cB_{\mathrm{reach}}\) has cells whose \(\dW\)-diameter is at most
\[
  \frac{\varepsilon}{L_V(\sigma)},
\]
where \(L_V(\sigma)\) is the induced value Lipschitz constant, then every
cell has decision diameter at most \(\varepsilon\). Such a cover is an FRR
at level \((\sigma,\varepsilon)\). Conversely, the policy sufficiency
theorem does not depend on how the cover was obtained: any FRR cover with
decision diameter at most \(\varepsilon\) supports a cell-constant policy
whose suboptimality is bounded by
\[
  \frac{2\varepsilon}{1-\gamma}.
\]
The cover may be obtained analytically, from a finite-POMDP simplex
construction, from a linear-Gaussian approximation, from particle-filter
rollouts with appropriate statistical margins, or by direct empirical
certification of within-cell \(Q^*\)-variation.

The framework sits upstream of robust MPC, stochastic MPC, safe adaptive
control, and control-barrier-function-based safety filters
\cite{mayne2016robust,mesbah2016stochastic,ames2019control}. Those methods
answer how to act once a state or belief representation has already been
chosen and uncertainty has been modeled. FRR addresses the prior
representation question: what belief resolution does the controller need in
the first place, given the physical noise floor and a target decision
tolerance? In this sense, FRR is not a competing controller. It supplies a
reliability-cell abstraction that can be combined with receding-horizon and
safety-critical execution methods.

The paper makes four main contributions. First, we separate predictive
observation distinguishability, action-execution uncertainty, and
belief-kernel smoothness as distinct objects in partially observed control.
We prove that additional independent sensor noise cannot increase
predictive observation distinguishability and show how actuation uncertainty
enters the planning model through an execution channel. Second, we identify
the controlled belief-transition kernel as the correct smoothness object for
dynamic programming and prove Lipschitz bounds for \(V^*\) and \(Q^*\) under
a reachable belief-kernel Lipschitz condition. Third, we define FRR covers
through bounded decision diameter and prove both a constructive cover
theorem and a policy sufficiency theorem for cell-constant policies. Fourth,
we introduce reliability entropy,
\[
  \cH_{\mathrm{FRR}}(\sigma,\varepsilon)
  =
  \log N_{\mathrm{FRR}}(\sigma,\varepsilon),
\]
as a rate-distortion-like measure of the minimum number of
decision-relevant belief messages required at tolerance \(\varepsilon\).
This quantity measures certified cell identity available to cell-constant
policies and provides a way to discuss usable policy capacity under physical
noise floors.

The theory also clarifies what should not be claimed. Although sensor
degradation reduces predictive observation distinguishability, sensor noise
alone does not universally imply that the number of required FRR cells
decreases. The reliability-cell count depends on information acquisition,
belief-kernel stability, executable action resolution, and value
sensitivity. In some regimes, degraded sensing may make belief uncertainty
more decision-critical and increase the number of cells required to certify
a fixed decision tolerance. In other regimes, actuation uncertainty may
collapse distinctions among commanded controls and reduce the
decision-relevant variation of \(Q^*\), allowing coarser cells. The
framework therefore separates observation distinguishability from decision
diameter instead of imposing a universal monotonicity claim.

Finally, we record complementary noise-floor lower bounds. The FRR guarantee
controls the additional loss caused by replacing the full belief space with
finitely many reliability cells, relative to the optimal policy for the same
noisy POMDP. It does not eliminate the physical consequences of noise itself.
We show through sensing, process-noise, and actuation lower-bound examples
that even a policy with unlimited representation capacity may be separated
from an ideal noiseless or disturbance-free oracle by a nonzero performance
gap. Thus, FRR should be interpreted as a representation theory under
physical uncertainty, not as a mechanism for removing irreducible sensing or
actuation limits.

The paper is organized as follows. Section~\ref{sec:setup} establishes the
partially observed stochastic system, the belief dynamics, the predictive
observation law, and the controlled belief-transition kernel.
Section~\ref{sec:distinguishability} develops noise-calibrated predictive
observation distinguishability and its data-processing property.
Section~\ref{sec:finite_pomdp} gives the finite-POMDP specialization,
including exact belief-simplex formulas, observation degradation, and
action-execution uncertainty. Section~\ref{sec:frr} defines FRR, proves
value and action-value Lipschitz bounds, constructs FRR covers, and
establishes policy sufficiency. Section~\ref{sec:entropy} introduces
reliability entropy, gives metric-entropy bounds, and explains usable
policy capacity and numerical certification. Section~\ref{sec:numerical}
presents exact finite-POMDP certification and sampled particle-filter
certification examples. Section~\ref{sec:limits} presents complementary
sensing, process, and actuation noise-floor lower bounds.
Section~\ref{sec:conclusion} concludes.

%============================================================================
\section{Problem Setup}
\label{sec:setup}
%============================================================================

\subsection{Nonlinear partially observed system}

Consider the discrete-time partially observed stochastic system
\begin{align}
  x_{t+1} &= f(x_t,u_t,w_t), \label{eq:dynamics}\\
  y_t &= h(x_t)+v_t, \label{eq:observation}
\end{align}
where $x_t\in\cX\subset\R^n$, $u_t\in\cU\subset\R^m$, and
$y_t\in\cY\subseteq\R^p$. The process noise is $w_t\sim \mu_w$ and the
observation noise is
\[
  v_t\sim \cN(0,\sigma^2 I_p),
\]
where $\sigma>0$ is the sensor noise floor.

The input \(u_t\) should be interpreted as a commanded control. Actuation or
execution uncertainty can be included in \(w_t\), or equivalently in the
controlled transition kernel induced by the pair \((u_t,w_t)\). Thus the
same notation covers both process disturbances and actuator noise; in the
observation-only portions of the analysis, \(\sigma\) denotes the sensor noise
floor.

\begin{assumption}[Compact reachable state set]
\label{ass:compact}
The set $\cX$ is compact with diameter $D_x<\infty$, and the dynamics map
$\cX$ into itself for all admissible controls and process-noise values in
the support of $\mu_w$:
\[
  f(\cX\times\cU\times\supp(\mu_w))\subseteq \cX .
\]
\end{assumption}

\begin{assumption}[Regularity of dynamics, observation, and reward]
\label{ass:regular}
The dynamics $f$ are Lipschitz in $x$ uniformly over $u$ and $w$, with
constant $L_f$. The observation map $h$ is Lipschitz with constant $L_h$.
The reward $r:\cX\times\cU\to\R$ is bounded and Lipschitz in $x$ uniformly
over $u$, with constant $L_r$:
\[
  |r(x,u)-r(x',u)|\le L_r\norm{x-x'}.
\]
\end{assumption}

\begin{assumption}[Compact controls and Bellman maximizers]
\label{ass:maximizers}
The admissible control set \(\cU\) is compact. Moreover, the reward and the
nominal belief-transition model are regular enough in \(u\) on
\(\cB_{\mathrm{reach}}\) that the Bellman maximizers used below exist. In
particular, for every \(b\in\cB_{\mathrm{reach}}\), the maps
\(u\mapsto Q^*(b,u)\) are upper semicontinuous, so that
\(\arg\max_{u\in\cU}Q^*(b,u)\) is nonempty.
\end{assumption}

The compactness assumption is not intended to replace stability analysis.
For physical systems with unbounded state spaces, $\cX$ should be read as
a compact control-invariant operating envelope, for example a verified
sublevel set or a bounded experimental domain.

\subsection{Belief space and Bayesian filter}

Let $\cB(\cX)$ denote the set of Borel probability measures on $\cX$,
equipped with the Wasserstein-1 metric
\[
  \dW(b,b')=
  \sup_{\mathrm{Lip}(\phi)\le 1}
  \left|\int_\cX \phi(x)\,b(dx)-\int_\cX\phi(x)\,b'(dx)\right|.
\]
The belief state is
\[
  b_t(\cdot)=\Pr(x_t\in\cdot\mid y_{0:t},u_{0:t-1}).
\]

For a belief $b$ and control $u$, the predicted belief is
\[
  b^-(A)=\int_\cX p(A\mid x,u)\,b(dx),
\]
where $p(\cdot\mid x,u)$ is the transition kernel induced by
$f(x,u,w)$ and $w\sim\mu_w$. Given an observation $y$, the Bayesian update
is
\begin{equation}
\begin{aligned}
  \frac{d\Phi_\sigma(b,u,y)}{db^-}(x)
  &=
  \frac{\ell_y(x)}
       {\int_\cX \ell_y(z)\,b^-(dz)}, \\
  \ell_y(x)&=\cN(y;h(x),\sigma^2 I_p).  
\end{aligned}
\label{eq:bayes}
\end{equation}

The map $\Phi_\sigma(b,u,y)$ is the deterministic filter update for a fixed
observation. Importantly, it is \emph{not} globally contractive on
$\cB(\cX)$ under Assumptions~\ref{ass:compact}--\ref{ass:regular}; the
present analysis does not require such a claim.

\subsection{Predictive observation law and belief-transition kernel}

For each $(b,u)$, define the predictive observation law
\[
  \cO_\sigma(\cdot\mid b,u)
  :=
  \int_\cX \cN(\cdot;h(x),\sigma^2 I_p)\,b^-(dx).
\]
This is the distribution of the next observation before it is received.

The controlled belief-transition kernel is the Markov kernel on
$\cB(\cX)$ defined by
\[
  \cK_\sigma(A\mid b,u)
  :=
  \int_{\cY}
  \mathbf 1_A\bigl(\Phi_\sigma(b,u,y)\bigr)
  \,\cO_\sigma(dy\mid b,u).
\]
Thus, $\cK_\sigma(\cdot\mid b,u)$ describes the random next belief after
the system generates an observation and the filter updates. The Bellman
operator depends on $\cK_\sigma$, not only on the fixed-observation map
$\Phi_\sigma$.

\subsection{Reachable belief set and value functions}

Let $\cB_{\mathrm{reach}}\subseteq\cB(\cX)$ denote a compact set of beliefs
that contains all beliefs reachable from the initial belief family under
admissible policies and observations of interest. A policy is a measurable
map $\pi:\cB_{\mathrm{reach}}\to\cU$. The expected reward under belief $b$
is
\[
  \rho(b,u)=\int_\cX r(x,u)\,b(dx).
\]
The value of a policy $\pi$ is
\[
  V^\pi(b)=
  \E_\pi\!\left[
  \sum_{t=0}^{\infty}\gamma^t\rho(b_t,u_t)
  \,\middle|\, b_0=b
  \right],
\]
where $\gamma\in(0,1)$. The optimal value satisfies
\begin{equation}
  V^*(b)=
  \sup_{u\in\cU}
  \left\{
  \rho(b,u)+\gamma
  \int_{\cB_{\mathrm{reach}}} V^*(b_+)\,
  \cK_\sigma(db_+\mid b,u)
  \right\}.
  \label{eq:bellman_kernel}
\end{equation}
The optimal action-value function is
\begin{equation}
  Q^*(b,u)=
  \rho(b,u)+\gamma
  \int_{\cB_{\mathrm{reach}}} V^*(b_+)\,
  \cK_\sigma(db_+\mid b,u).
  \label{eq:qstar}
\end{equation}

%============================================================================
\section{Noise-Calibrated Belief Distinguishability}
\label{sec:distinguishability}
%============================================================================

Sensor noise determines not only the variance of an estimate, but also the
resolution at which different beliefs can be distinguished through future
measurements. This section introduces a distinguishability-level object that
captures this idea without asserting false contraction of the Bayesian update.

\begin{definition}[Predictive observation distinguishability]
\label{def:obs_metric}
For fixed $\sigma$ and $u$, define the predictive observation
distinguishability between two beliefs $b,b'\in\cB(\cX)$ by
\[
  d_{\mathrm{obs},\sigma}^u(b,b')
  :=
  \TV\!\left(
    \cO_\sigma(\cdot\mid b,u),
    \cO_\sigma(\cdot\mid b',u)
  \right),
\]
where $\TV(P,Q):=\sup_A |P(A)-Q(A)|$ is total variation distance
$(\TV(P,Q)=\frac12\sum_{y\in\cY}|P(y)-Q(y)|$ when $\cY$ is finite).
The control-uniform version is
\[
  d_{\mathrm{obs},\sigma}(b,b')
  :=
  \sup_{u\in\cU} d_{\mathrm{obs},\sigma}^u(b,b').
\]
\end{definition}

This quantity measures how distinguishable two beliefs are through the
sensor channel after one prediction step. If two beliefs induce nearly the
same predictive observation law, no decision rule based only on the next
observation can reliably separate them.

\begin{proposition}[Sensor degradation cannot increase distinguishability]
\label{prop:data_processing}
Suppose the observation model with noise level $\sigma_2$ is obtained from
the observation model with noise level $\sigma_1$ by adding independent
Gaussian noise:
\[
  y_{\sigma_2} = y_{\sigma_1} + \eta,
  \qquad
  \eta \sim \mathcal N(0,\tau^2 I),
  \qquad
  \sigma_2^2=\sigma_1^2+\tau^2 .
\]
For this continuous Gaussian degradation result, regard the predictive
observation laws as probability measures on $\mathbb R^p$. Then, for all
$b,b'\in\cB(\cX)$ and $u\in\cU$,
\[
  d_{\mathrm{obs},\sigma_2}^u(b,b')
  \le
  d_{\mathrm{obs},\sigma_1}^u(b,b').
\]
\end{proposition}

\begin{proof}
Let
\[
  P_1=\cO_{\sigma_1}(\cdot\mid b,u),
  \qquad
  Q_1=\cO_{\sigma_1}(\cdot\mid b',u)
\]
be the predictive observation distributions at noise level $\sigma_1$.
Since the sensor with noise level $\sigma_2$ is obtained by adding
independent Gaussian noise to the sensor with noise level $\sigma_1$, define
the additive-noise Markov kernel
\[
  K_\tau(A\mid y)
  :=
  \int_A \mathcal N(y';y,\tau^2 I)\,dy'
\]
for every measurable set \(A\subseteq\mathbb R^p\). Thus \(K_\tau(A\mid y)\)
is the probability that the degraded observation falls in \(A\), given the
lower-noise observation \(y\).

The predictive observation distributions at noise level \(\sigma_2\) are
therefore
\[
\begin{aligned}
  P_2(A)&=\int K_\tau(A\mid y)\,P_1(dy), \\
  Q_2(A)&=\int K_\tau(A\mid y)\,Q_1(dy).
\end{aligned}
\]
Equivalently, \(P_2=P_1K_\tau\) and \(Q_2=Q_1K_\tau\).

Total variation distance contracts under a common Markov kernel. Indeed,
using the variational characterization of total variation,
\[
\begin{aligned}
  \TV(P_2,Q_2)
  &=
  \sup_A
  \left|
  \int K_\tau(A\mid y)\,(P_1-Q_1)(dy)
  \right| \\
  &\le
  \sup_{0\le f\le 1}
  \left|
  \int f(y)\,(P_1-Q_1)(dy)
  \right| \\
  &=
  \TV(P_1,Q_1).
\end{aligned}
\]
Therefore,
\[
  d_{\mathrm{obs},\sigma_2}^u(b,b')
  =
  \TV(P_2,Q_2)
  \le
  \TV(P_1,Q_1)
  =
  d_{\mathrm{obs},\sigma_1}^u(b,b'),
\]
as claimed.
\end{proof}

Proposition~\ref{prop:data_processing} is the precise form of the intuition
that a noisier sensor cannot create additional observational
distinguishability. It does not imply that the Bayesian filter is globally
contractive, nor that all belief differences become decision-irrelevant. It
only says that the observation channel itself becomes less capable of
distinguishing beliefs as additional independent noise is appended.

\begin{remark}[Linear-Gaussian specialization]
\label{rem:linear_gaussian}
In the linear-Gaussian case
\[
  x_{t+1}=Ax_t+Bu_t+w_t,\qquad y_t=Cx_t+v_t,
\]
with Gaussian beliefs $b=\cN(\mu,\Sigma_\sigma)$ and
$b'=\cN(\mu',\Sigma_\sigma)$, the predictive observation laws are Gaussian
with common covariance
\[
  S_\sigma=C\Sigma_\sigma C^\top+\sigma^2 I.
\]
The symmetrized KL divergence reduces to the Mahalanobis quantity
\[
  (\mu-\mu')^\top C^\top S_\sigma^{-1}C(\mu-\mu').
\]
Thus, in this special case, the sensor-induced distinguishability quantity
is explicit and decreases as the observation covariance increases.
\end{remark}

%============================================================================
\section{Finite-State POMDPs as an Exact FRR Model}
\label{sec:finite_pomdp}
%============================================================================

The preceding definitions were stated for a general belief space. A
finite-state POMDP gives the cleanest exact instance of the same framework.
In this setting, the physical state space is finite, but the belief space is
still continuous: beliefs are probability vectors on the simplex
$\Delta^{n-1}$. The advantage is that this continuous belief space is
finite-dimensional. The prediction step is linear, and the Bayesian
correction step is a normalized positive linear map. Thus, the finite POMDP
case preserves the key nonlinear feature of filtering while avoiding the
functional-analytic complications associated with probability measures on
continuous state spaces.

\subsection{Belief simplex, prediction, and Bayes correction}

Let
\[
  \cS=\{1,\dots,n\},\qquad
  \cA=\{1,\dots,m\},\qquad
  \cZ=\{1,\dots,k\}
\]
be finite state, action, and observation sets. A belief is a vector
\[
  b\in\DeltaS:=\left\{b\in\R^n_{\ge 0}:\mathbf 1^\top b=1\right\}.
\]
For action $a\in\cA$, let $P_a\in\R^{n\times n}$ be the transition matrix,
with $[P_a]_{ij}=\Pr(s_{t+1}=j\mid s_t=i,a_t=a)$. The predicted belief is
linear:
\begin{equation}
  \bar b_a=P_a^\top b .
  \label{eq:finite_prediction}
\end{equation}
Let $O_{a,\sigma}(z\mid s')$ denote the observation model under sensor noise
level $\sigma$, and define
\[
  D_{a,z,\sigma}
  :=\diag\bigl(O_{a,\sigma}(z\mid 1),\dots,O_{a,\sigma}(z\mid n)\bigr).
\]
The probability of observing $z$ from belief $b$ after action $a$ is
\begin{equation}
  p_\sigma(z\mid b,a)
  =
  \mathbf 1^\top D_{a,z,\sigma}P_a^\top b .
  \label{eq:finite_obs_prob}
\end{equation}
Whenever $p_\sigma(z\mid b,a)>0$, the Bayesian update is
\begin{equation}
  \tau_\sigma(b,a,z)
  =
  \frac{D_{a,z,\sigma}P_a^\top b}
       {\mathbf 1^\top D_{a,z,\sigma}P_a^\top b}.
  \label{eq:finite_bayes}
\end{equation}
Equation~\eqref{eq:finite_bayes} is not linear in $b$ because of the
normalization denominator. It is a normalized linear, or projective, map on
the simplex. This is the finite-state analogue of the continuous Bayesian
filter map $\Phi_\sigma(b,u,y)$.

The corresponding controlled belief-transition kernel is the finite sum
\begin{equation}
  \cK_\sigma(A\mid b,a)
  =
  \sum_{z\in\cZ}
  p_\sigma(z\mid b,a)\,
  \mathbf 1_A\bigl(\tau_\sigma(b,a,z)\bigr),
  \label{eq:finite_belief_kernel}
\end{equation}
for every measurable $A\subseteq\DeltaS$. This is exactly the same object as
$\cK_\sigma$ in the nonlinear theory, specialized to a finite observation
set.

\begin{remark}[Finite MDPs versus finite POMDPs]
For a fully observed finite MDP, probability distributions evolve linearly
under a fixed action: $b_{t+1}=P_a^\top b_t$. For a finite POMDP, the
prediction step remains linear, but the Bayesian correction step
\eqref{eq:finite_bayes} is normalized and therefore nonlinear. This is why
tabular MDPs are linear on distributions, whereas tabular POMDP belief
dynamics are nonlinear on the simplex.
\end{remark}

\subsection{Observation distinguishability on the simplex}

The predictive observation distribution is the vector
\[
  p_\sigma(\cdot\mid b,a)
  =
  O_{a,\sigma}^\top P_a^\top b,
\]
where $[O_{a,\sigma}]_{s'z}=O_{a,\sigma}(z\mid s')$. Therefore the finite
POMDP version of Definition~\ref{def:obs_metric} becomes
\begin{equation}
  d_{\mathrm{obs},\sigma}(b,b')
  :=
  \max_{a\in\cA}
  \frac12
  \left\|O_{a,\sigma}^\top P_a^\top(b-b')\right\|_1 .
  \label{eq:finite_obs_metric}
\end{equation}
This distinguishability pseudometric is explicitly sensor-dependent. If two
beliefs induce nearly the same predictive observation distribution, then no
one-step decision rule using the next observation can reliably tell them
apart.

\begin{theorem}[Sensor degradation contracts observation distinguishability]
\label{thm:finite_sensor_data_processing}
Suppose the noisier observation channel is obtained by passing the less noisy
observation through a stochastic degradation matrix $G$. That is, for each
action $a\in\cA$,
\[
  O_{a,\sigma_2}=O_{a,\sigma_1}G,
\]
where
\[
  G_{z\tilde z}=\Pr(\tilde z\mid z),
  \,\,
  G_{z\tilde z}\ge 0,
  \,\,
  \sum_{\tilde z\in\cZ}G_{z\tilde z}=1
  \,\, \text{for every } z\in\cZ.
\]
Equivalently,
\[
  O_{a,\sigma_2}(\tilde z\mid s)
  =
  \sum_{z\in\cZ}
  O_{a,\sigma_1}(z\mid s)G_{z\tilde z}.
\]
Then
\[
  d_{\mathrm{obs},\sigma_2}(b,b')
  \le
  d_{\mathrm{obs},\sigma_1}(b,b')
  \qquad
  \forall b,b'\in\DeltaS .
\]
\end{theorem}

\begin{proof}
Fix an action $a\in\cA$. Let
\[
  \bar b_a=P_a^\top b,
  \qquad
  \bar b'_a=P_a^\top b'
\]
be the predicted beliefs after applying action $a$. Define the predictive
observation distributions under the less noisy channel by
\[
\begin{aligned}
  q_1(z)
  &:={}
  \sum_{s\in\cS}
  O_{a,\sigma_1}(z\mid s)\bar b_a(s), \\
  q'_1(z)
  &:={}
  \sum_{s\in\cS}
  O_{a,\sigma_1}(z\mid s)\bar b'_a(s).
\end{aligned}
\]
Since $O_{a,\sigma_2}=O_{a,\sigma_1}G$, the predictive observation
distributions under the noisier channel satisfy
\[
  q_2(\tilde z)
  =
  \sum_{z\in\cZ}q_1(z)G_{z\tilde z},
  \qquad
  q'_2(\tilde z)
  =
  \sum_{z\in\cZ}q'_1(z)G_{z\tilde z}.
\]
Therefore,
\begin{align*}
  \TV(q_2,q'_2)
  &=
  \frac12
  \sum_{\tilde z\in\cZ}
  \left|q_2(\tilde z)-q'_2(\tilde z)\right| \\
  &=
  \frac12
  \sum_{\tilde z\in\cZ}
  \left|
    \sum_{z\in\cZ}
    \bigl(q_1(z)-q'_1(z)\bigr)G_{z\tilde z}
  \right| \\
  &\le
  \frac12
  \sum_{\tilde z\in\cZ}
  \sum_{z\in\cZ}
  \left|q_1(z)-q'_1(z)\right|G_{z\tilde z} \\
  &=
  \frac12
  \sum_{z\in\cZ}
  \left|q_1(z)-q'_1(z)\right|
  \sum_{\tilde z\in\cZ}G_{z\tilde z} \\
  &=
  \frac12
  \sum_{z\in\cZ}
  \left|q_1(z)-q'_1(z)\right|
  =
  \TV(q_1,q'_1).
\end{align*}
The inequality uses the triangle inequality and the nonnegativity of
$G_{z\tilde z}$; the penultimate equality uses
$\sum_{\tilde z}G_{z\tilde z}=1$. Hence, for this fixed action,
\[
  \frac12
  \left\|O_{a,\sigma_2}^\top P_a^\top(b-b')\right\|_1
  \le
  \frac12
  \left\|O_{a,\sigma_1}^\top P_a^\top(b-b')\right\|_1 .
\]
Taking the maximum over $a\in\cA$ gives the claimed contraction of
$d_{\mathrm{obs},\sigma}$.
\end{proof}

\begin{proposition}[Explicit sensor degradation scaling]
\label{prop:finite_mixture_degradation}
Consider the one-parameter finite-observation model
\[
  O_{a,\lambda}(z\mid s)
  =
  (1-\lambda)O_{a,0}(z\mid s)+\lambda q_a(z),
  \qquad \lambda\in[0,1],
\]
where $q_a$ is a state-independent observation distribution. Then, for each
fixed action $a$,
\[
  d_{\mathrm{obs},\lambda}^a(b,b')
  =
  (1-\lambda)d_{\mathrm{obs},0}^a(b,b').
\]
Consequently, since the same $\lambda$ is used for all actions,
\[
  d_{\mathrm{obs},\lambda}(b,b')
  =
  (1-\lambda)d_{\mathrm{obs},0}(b,b').
\]
\end{proposition}

\begin{proof}
Fix an action $a$. Let
\[
  \bar b_a=P_a^\top b,
  \qquad
  \bar b'_a=P_a^\top b'
\]
be the predicted beliefs. The predictive observation distribution under the
undegraded sensor is
\[
  p_0(z\mid b,a)
  =
  \sum_{s\in\cS}O_{a,0}(z\mid s)\bar b_a(s),
\]
and similarly for $p_0(z\mid b',a)$. Under the degraded sensor,
\begin{align*}
  p_\lambda(z\mid b,a)
  &=
  \sum_{s\in\cS}
  O_{a,\lambda}(z\mid s)\bar b_a(s) \\
  &=
  \sum_{s\in\cS}
  \bigl[(1-\lambda)O_{a,0}(z\mid s)+\lambda q_a(z)\bigr]\bar b_a(s) \\
  &=
  (1-\lambda)p_0(z\mid b,a)+\lambda q_a(z),
\end{align*}
because $\sum_s\bar b_a(s)=1$. Likewise,
\[
  p_\lambda(z\mid b',a)
  =
  (1-\lambda)p_0(z\mid b',a)+\lambda q_a(z).
\]
Therefore the state-independent term cancels:
\[
  p_\lambda(z\mid b,a)-p_\lambda(z\mid b',a)
  =
  (1-\lambda)
  \bigl[p_0(z\mid b,a)-p_0(z\mid b',a)\bigr].
\]
Since $1-\lambda\ge 0$,
\begin{align*}
  d_{\mathrm{obs},\lambda}^a(b,b')
  &=
  \frac12
  \sum_{z\in\cZ}
  \left|p_\lambda(z\mid b,a)-p_\lambda(z\mid b',a)\right| \\
  &=
  (1-\lambda)
  \frac12
  \sum_{z\in\cZ}
  \left|p_0(z\mid b,a)-p_0(z\mid b',a)\right| \\
  &=
  (1-\lambda)d_{\mathrm{obs},0}^a(b,b').
\end{align*}
Taking the maximum over $a\in\cA$ gives
\begin{align*}
  d_{\mathrm{obs},\lambda}(b,b')
  &=
  \max_{a\in\cA}d_{\mathrm{obs},\lambda}^a(b,b') \\
  &=
  (1-\lambda)
  \max_{a\in\cA}d_{\mathrm{obs},0}^a(b,b') \\
  &=
  (1-\lambda)d_{\mathrm{obs},0}(b,b').
\end{align*}
\end{proof}

\subsection{Action-execution uncertainty}
\label{subsec:action_execution_uncertainty}

The finite-POMDP model can also represent actuation uncertainty. Let
$a\in\cA$ denote the commanded action and let $\tilde a\in\cA$ denote the
action actually executed by the system. An action-execution channel
$G_{\sigu}$ is defined by
\[
  G_{\sigu}(\tilde a\mid a)
  =
  \Pr(\text{executed action }\tilde a
  \mid
  \text{commanded action }a),
\]
where $\sigu$ denotes the actuation noise floor. This channel changes the
POMDP associated with a commanded action. In the general case, the executed
action determines both the transition and the observation model, so the
predictive observation probability becomes
\begin{equation}
  p_{\sigy,\sigu}(z\mid b,a)
  =
  \sum_{\tilde a\in\cA}
  G_{\sigu}(\tilde a\mid a)
  \mathbf 1^\top D_{\tilde a,z,\sigy}P_{\tilde a}^\top b .
  \label{eq:finite_action_uncertain_obs_prob}
\end{equation}
The corresponding Bayesian update is
\begin{equation}
  \tau_{\sigy,\sigu}(b,a,z)
  =
  \frac{
  \sum_{\tilde a\in\cA}
  G_{\sigu}(\tilde a\mid a)
  D_{\tilde a,z,\sigy}P_{\tilde a}^\top b}
  {p_{\sigy,\sigu}(z\mid b,a)} .
  \label{eq:finite_action_uncertain_bayes}
\end{equation}
The belief-transition kernel is obtained by summing
\eqref{eq:finite_action_uncertain_bayes} over observations with probabilities
\eqref{eq:finite_action_uncertain_obs_prob}. Thus no new FRR definition is
needed: action uncertainty simply changes the controlled belief-transition
kernel and the resulting action-value function.

The same channel also averages immediate rewards. The reward associated with
the commanded action $a$ is
\[
  \bar r_{\sigu}(s,a)
  =
  \sum_{\tilde a\in\cA}
  G_{\sigu}(\tilde a\mid a) r(s,\tilde a).
\]
If the observation model depends only on the next state and not on the
executed action, then the effective commanded-action transition is
\[
  \bar P_{a,\sigu}
  =
  \sum_{\tilde a\in\cA}
  G_{\sigu}(\tilde a\mid a)P_{\tilde a},
\]
and the prediction step becomes
\[
  \bar b_{a,\sigu}=\bar P_{a,\sigu}^\top b.
\]
This is the finite-state analogue of including actuator noise in the process
noise of the nonlinear model.

The command-level distinguishability induced by an execution channel is
\begin{equation}
  d_{\mathrm{exec},\sigu}(a,a')
  :=
  \frac12
  \sum_{\tilde a\in\cA}
  \left|
    G_{\sigu}(\tilde a\mid a)-G_{\sigu}(\tilde a\mid a')
  \right|.
  \label{eq:execution_distinguishability}
\end{equation}
This quantity measures how reliably the physical plant can distinguish two
command labels at the execution layer.

\begin{theorem}[Actuator degradation contracts command distinguishability]
\label{thm:finite_actuator_data_processing}
Let \(G_{\sigma_{u,1}}\) and \(G_{\sigma_{u,2}}\) be two
action-execution channels. Suppose \(G_{\sigma_{u,2}}\) is obtained by
further degrading \(G_{\sigma_{u,1}}\) through a stochastic matrix \(H\) on
the executed-action alphabet, namely
\[
  G_{\sigma_{u,2}}(\hat a\mid a)
  =
  \sum_{\tilde a\in\cA}
  G_{\sigma_{u,1}}(\tilde a\mid a)H_{\tilde a\hat a},
\]
where
\[
  H_{\tilde a\hat a}\ge 0,
  \qquad
  \sum_{\hat a\in\cA}H_{\tilde a\hat a}=1
  \quad \text{for every } \tilde a\in\cA .
\]
Define the execution distinguishability between two commanded actions by
\[
  d_{\mathrm{exec},\sigma_u}(a,a')
  :=
  \frac12
  \sum_{\tilde a\in\cA}
  \left|
    G_{\sigma_u}(\tilde a\mid a)
    -
    G_{\sigma_u}(\tilde a\mid a')
  \right|.
\]
Then, for every pair of commanded actions \(a,a'\in\cA\),
\[
  d_{\mathrm{exec},\sigma_{u,2}}(a,a')
  \le
  d_{\mathrm{exec},\sigma_{u,1}}(a,a').
\]
In particular, for two actions with the symmetric execution channel
\[
  G_\alpha(\tilde a\mid a)
  =
  \begin{cases}
  1-\alpha, & \tilde a=a,\\
  \alpha, & \tilde a\neq a,
  \end{cases}
  \qquad
  0\le \alpha\le \frac12,
\]
the two command labels satisfy
\[
  d_{\mathrm{exec},\alpha}(a_1,a_2)=1-2\alpha.
\]
Thus the two commands become execution-indistinguishable as
\(\alpha\to 1/2\).
\end{theorem}

\begin{proof}
Fix two commanded actions \(a,a'\in\cA\). For each
\(\tilde a\in\cA\), define
\[
  h(\tilde a)
  :=
  G_{\sigma_{u,1}}(\tilde a\mid a)
  -
  G_{\sigma_{u,1}}(\tilde a\mid a').
\]
Using the degradation relation,
\[
\begin{aligned}
  &G_{\sigma_{u,2}}(\hat a\mid a)
  -
  G_{\sigma_{u,2}}(\hat a\mid a') \\
  &\quad =
  \sum_{\tilde a\in\cA}
  \left[
    G_{\sigma_{u,1}}(\tilde a\mid a)
    -
    G_{\sigma_{u,1}}(\tilde a\mid a')
  \right]H_{\tilde a\hat a} \\
  &\quad =
  \sum_{\tilde a\in\cA}
  h(\tilde a)H_{\tilde a\hat a}.
\end{aligned}
\]
Therefore,
\[
\begin{aligned}
  d_{\mathrm{exec},\sigma_{u,2}}(a,a')
  &=
  \frac12
  \sum_{\hat a\in\cA}
  \left|
    G_{\sigma_{u,2}}(\hat a\mid a)
    -
    G_{\sigma_{u,2}}(\hat a\mid a')
  \right| \\
  &=
  \frac12
  \sum_{\hat a\in\cA}
  \left|
    \sum_{\tilde a\in\cA}
    h(\tilde a)H_{\tilde a\hat a}
  \right| \\
  &\le
  \frac12
  \sum_{\hat a\in\cA}
  \sum_{\tilde a\in\cA}
  |h(\tilde a)|H_{\tilde a\hat a} \\
  &=
  \frac12
  \sum_{\tilde a\in\cA}
  |h(\tilde a)|
  \sum_{\hat a\in\cA}H_{\tilde a\hat a} \\
  &=
  \frac12
  \sum_{\tilde a\in\cA}
  |h(\tilde a)| \\
  &=
  d_{\mathrm{exec},\sigma_{u,1}}(a,a').
\end{aligned}
\]
This proves that further actuator degradation cannot increase command
distinguishability.

For the two-action symmetric channel, let
\(\cA=\{a_1,a_2\}\). Then
\[
  G_\alpha(\cdot\mid a_1)=(1-\alpha,\alpha),
  \qquad
  G_\alpha(\cdot\mid a_2)=(\alpha,1-\alpha),
\]
with coordinates ordered as \((a_1,a_2)\). Hence
\[
\begin{aligned}
  d_{\mathrm{exec},\alpha}(a_1,a_2)
  &=
  \frac12
  \left(
  |(1-\alpha)-\alpha|
  +
  |\alpha-(1-\alpha)|
  \right) \\
  &=
  |1-2\alpha| \\
  &=
  1-2\alpha,
\end{aligned}
\]
because \(0\le \alpha\le 1/2\). Therefore
\[
  d_{\mathrm{exec},\alpha}(a_1,a_2)\to 0
  \qquad \text{as} \qquad
  \alpha\to \frac12 .
\]
Thus the two command labels become indistinguishable in execution.
\end{proof}

\begin{remark}[Scope of the actuator contraction]
Theorem~\ref{thm:finite_actuator_data_processing} is a statement about
physical command distinguishability at the execution channel. It does not by
itself assert that action-value gaps are monotone in the actuation noise
parameter for every POMDP, because changing the executed-action mixture also
changes the transition, observation, and reward models. The value consequence
is model-dependent and is certified by the decision-diameter condition. In
the limiting case where two commanded actions induce identical effective
rewards, transitions, and observation laws, their corresponding action values
are identical for every belief.
\end{remark}

\subsection{Exact action-value Lipschitzness}

Finite POMDPs also give an exact and simple way to obtain a decision-diameter
bound. Let the reward vector for action $a$ be $r_a\in\R^n$, with
$|r_a(s)|\le R_{\max}$ for all $s$ and $a$.

\begin{theorem}[Finite-POMDP action-value Lipschitz bound]
\label{thm:finite_q_lipschitz}
For a discounted finite POMDP with $|r_a(s)|\le R_{\max}$,
\[
  \sup_{a\in\cA}|Q^*(b,a)-Q^*(b',a)|
  \le
  L_Q^\Delta\,\norm{b-b'}_1,
  \quad
  L_Q^\Delta:=\frac{R_{\max}}{1-\gamma}.
\]
Consequently, any cell $C\subseteq\DeltaS$ satisfying
\[
  \diam_1(C):=\sup_{b,b'\in C}\norm{b-b'}_1
  \le
  \frac{\varepsilon}{L_Q^\Delta}
\]
has decision diameter $\Delta_Q(C)\le\varepsilon$.
\end{theorem}

\begin{proof}
Fix an action $a\in\cA$. For any continuation policy $\pi$, define the
action-conditioned $\alpha$-vector $\alpha_a^\pi\in\mathbb R^{|\cS|}$ by
\[
  \alpha_a^\pi(s)
  :=
  \mathbb E^\pi\!\left[
    \sum_{t=0}^{\infty}\gamma^t r_{a_t}(s_t)
    \,\middle|\,
    s_0=s,\ a_0=a
  \right],
  \qquad s\in\cS .
\]
Thus, $\alpha_a^\pi(s)$ is the expected discounted return obtained when the
hidden state starts at $s$, the first action is fixed to be $a$, and the
continuation policy $\pi$ is used thereafter.

By conditioning on the initial hidden state, the action-value function of
policy $\pi$ is linear in the initial belief:
\[
  Q^\pi(b,a)
  =
  \sum_{s\in\cS}b(s)\alpha_a^\pi(s)
  =
  (\alpha_a^\pi)^\top b .
\]
Since $|r_a(s)|\le R_{\max}$ for every state-action pair, every discounted
return is bounded in magnitude by
\[
  \sum_{t=0}^{\infty}\gamma^t R_{\max}
  =
  \frac{R_{\max}}{1-\gamma}.
\]
Therefore, for every continuation policy $\pi$,
\[
  \norm{\alpha_a^\pi}_\infty
  \le
  \frac{R_{\max}}{1-\gamma}.
\]

Let
\[
  \Gamma_a
  :=
  \{\alpha_a^\pi:\pi\text{ is an admissible continuation policy}\}.
\]
Then the optimal action-value function with first action fixed to $a$ can be
written as
\[
  Q^*(b,a)
  =
  \sup_{\alpha\in\Gamma_a}\alpha^\top b .
\]
For any beliefs $b,b'\in\DeltaS$,
\begin{align*}
  Q^*(b,a)-Q^*(b',a)
  &=
  \sup_{\alpha\in\Gamma_a}\alpha^\top b
  -
  \sup_{\alpha\in\Gamma_a}\alpha^\top b' \\
  &\le
  \sup_{\alpha\in\Gamma_a}
  \left[\alpha^\top b-\alpha^\top b'\right] \\
  &=
  \sup_{\alpha\in\Gamma_a}\alpha^\top(b-b') \\
  &\le
  \sup_{\alpha\in\Gamma_a}
  \norm{\alpha}_\infty\norm{b-b'}_1 \\
  &\le
  \frac{R_{\max}}{1-\gamma}\norm{b-b'}_1 .
\end{align*}
Swapping $b$ and $b'$ gives
\[
  Q^*(b',a)-Q^*(b,a)
  \le
  \frac{R_{\max}}{1-\gamma}\norm{b-b'}_1 .
\]
Hence
\[
  |Q^*(b,a)-Q^*(b',a)|
  \le
  \frac{R_{\max}}{1-\gamma}\norm{b-b'}_1 .
\]
Since the bound is independent of $a$, taking the supremum over $a\in\cA$
gives
\[
  \sup_{a\in\cA}|Q^*(b,a)-Q^*(b',a)|
  \le
  L_Q^\Delta\norm{b-b'}_1,
  \quad
  L_Q^\Delta:=\frac{R_{\max}}{1-\gamma}.
\]

Now let $C\subseteq\DeltaS$ satisfy
\[
  \diam_1(C)
  :=
  \sup_{b,b'\in C}\norm{b-b'}_1
  \le
  \frac{\varepsilon}{L_Q^\Delta}.
\]
Then
\[
\begin{aligned}
  \Delta_Q(C)
  &=
  \sup_{b,b'\in C}\sup_{a\in\cA}
  |Q^*(b,a)-Q^*(b',a)| \\
  &\le
  L_Q^\Delta\diam_1(C) \\
  &\le
  \varepsilon .
\end{aligned}
\]
Thus $C$ has decision diameter at most $\varepsilon$.
\end{proof}

Theorem~\ref{thm:finite_q_lipschitz} gives a fully rigorous finite-state
version of FRR without requiring any contraction of the Bayesian update. It
uses only the geometry of the belief simplex and the bounded-return structure
of discounted POMDPs.

\subsection{Covering consequences}

\begin{corollary}[Simplex covering bound]
\label{cor:finite_simplex_cover}
Let $N_\Delta(\varepsilon)$ be the number of $\ell_1$-diameter cells needed
to construct an FRR on $\DeltaS$ using
Theorem~\ref{thm:finite_q_lipschitz}. Then there exists an explicit covering
satisfying
\[
  N_\Delta(\varepsilon)
  \le
  \left(
    1+\frac{2(n-1)L_Q^\Delta}{\varepsilon}
  \right)^{n-1}.
\]
In particular,
\[
  N_\Delta(\varepsilon)
  \lesssim_n
  \left(
    \frac{2L_Q^\Delta}{\varepsilon}
  \right)^{n-1},
\]
where the hidden constant depends only on $n$.
\end{corollary}

\begin{proof}
Let
\[
  r:=\frac{\varepsilon}{L_Q^\Delta}.
\]
By Theorem~\ref{thm:finite_q_lipschitz}, any set $C\subseteq\DeltaS$ with
$\diam_1(C)\le r$ has decision diameter at most $\varepsilon$. It therefore
suffices to construct a cover of $\DeltaS$ by sets of $\ell_1$-diameter at
most $r$.

If $n=1$, the simplex consists of a single point and the claim is trivial.
Assume $n\ge 2$. Parameterize $\DeltaS$ by its first $n-1$ coordinates:
\[
  T
  :=
  \left\{
    x\in\R^{n-1}_{\ge 0}:\sum_{i=1}^{n-1}x_i\le 1
  \right\},
\]
with the map
\[
  \phi(x)
  :=
  \bigl(x_1,\dots,x_{n-1},1-\sum_{i=1}^{n-1}x_i\bigr).
\]
For any $x,y\in T$,
\begin{align*}
  \norm{\phi(x)-\phi(y)}_1
  &=
  \sum_{i=1}^{n-1}|x_i-y_i|
  +
  \left|
    \sum_{i=1}^{n-1}(x_i-y_i)
  \right| \\
  &\le
  2\sum_{i=1}^{n-1}|x_i-y_i| \\
  &\le
  2(n-1)\norm{x-y}_\infty .
\end{align*}

Cover the cube $[0,1]^{n-1}$ by axis-aligned cubes of side length
\[
  h:=\frac{r}{2(n-1)}.
\]
Intersect this grid with $T$, and map the resulting cells into $\DeltaS$ by
$\phi$. If $b=\phi(x)$ and $b'=\phi(y)$ lie in the image of the same grid
cell, then $\norm{x-y}_\infty\le h$, and therefore
\[
  \norm{b-b'}_1
  =
  \norm{\phi(x)-\phi(y)}_1
  \le
  2(n-1)h
  =
  r.
\]
Thus every cell in the constructed covering has $\ell_1$-diameter at most
$r$.

The number of grid intervals needed along each coordinate is at most
$\lceil 1/h\rceil$. Since $\lceil x\rceil\le 1+x$ for $x\ge 0$, the number
of grid cells is bounded by
\[
  \left\lceil \frac{1}{h}\right\rceil^{n-1}
  \le
  \left(1+\frac{1}{h}\right)^{n-1}
  =
  \left(1+\frac{2(n-1)}{r}\right)^{n-1}.
\]
Substituting $r=\varepsilon/L_Q^\Delta$ gives
\[
  N_\Delta(\varepsilon)
  \le
  \left(
    1+\frac{2(n-1)L_Q^\Delta}{\varepsilon}
  \right)^{n-1}.
\]
The asymptotic scaling follows immediately, with constants depending only on
$n$.
\end{proof}

The preceding covering bound is purely geometric: it uses the raw
$\ell_1$-geometry of the belief simplex. A sensor-induced covering is more
model-specific. It requires an observability-type condition, because
$d_{\mathrm{obs},\sigma}$ is only a pseudometric and may collapse belief
directions that still matter for reward or future value.

The following optional assumption provides one sufficient condition under
which the \(\ell_1\)-Lipschitz bound can be transferred to the
sensor-induced pseudometric.

\begin{assumption}[Observation-to-belief stability]
\label{ass:obs_to_belief_stability}
For a fixed sensor noise level \(\sigma\), there exists a finite constant
\(C_{\mathrm{obs}}(\sigma)<\infty\) such that
\[
  \norm{b-b'}_1
  \le
  C_{\mathrm{obs}}(\sigma)d_{\mathrm{obs},\sigma}(b,b')
  \qquad
  \forall b,b'\in\mathcal B_{\mathrm{reach}},
\]
where \(\mathcal B_{\mathrm{reach}}\subseteq\DeltaS\) denotes the reachable
belief set under the POMDP dynamics and admissible policies.
\end{assumption}

\begin{remark}[Interpretation of Assumption~\ref{ass:obs_to_belief_stability}]
Assumption~\ref{ass:obs_to_belief_stability} is an observability-type
condition. It is not required for the basic \(\ell_1\)-based FRR guarantee
in Theorem~\ref{thm:finite_q_lipschitz}; it is only needed when one wants to
replace \(\ell_1\)-diameter by diameter in the sensor-induced
distinguishability pseudometric \(d_{\mathrm{obs},\sigma}\).

To see its meaning, define
\[
  M_{a,\sigma}:=O_{a,\sigma}^{\top}P_a^{\top}.
\]
Then
\[
  d_{\mathrm{obs},\sigma}(b,b')
  =
  \max_{a\in\cA}
  \frac12\norm{M_{a,\sigma}(b-b')}_1 .
\]
Since \(b-b'\) lies in the tangent subspace
\[
  T:=\{x\in\mathbb R^n:\mathbf 1^\top x=0\},
\]
the assumption holds on the full simplex if and only if there exists
\(\kappa_{\mathrm{obs}}(\sigma)>0\) such that
\[
  \max_{a\in\cA}
  \frac12\norm{M_{a,\sigma}x}_1
  \ge
  \kappa_{\mathrm{obs}}(\sigma)\norm{x}_1
  \qquad
  \forall x\in T.
\]
In that case one may take
\[
  C_{\mathrm{obs}}(\sigma)
  =
  \frac{1}{\kappa_{\mathrm{obs}}(\sigma)}.
\]
Equivalently, the collection of observation-prediction maps
\(\{M_{a,\sigma}\}_{a\in\cA}\) must be injective on the belief-difference
subspace \(T\). If there exists a nonzero \(x\in T\) such that
\[
  M_{a,\sigma}x=0
  \qquad
  \forall a\in\cA,
\]
then two distinct beliefs can induce identical predictive observation laws
for every action, so \(d_{\mathrm{obs},\sigma}\) cannot control
\(\norm{b-b'}_1\), and the assumption fails.

This condition is therefore strongest when the sensor is highly informative
and becomes weaker as the sensor is degraded. Under severe sensor noise,
dropout, or state-independent observations, the constant
\(C_{\mathrm{obs}}(\sigma)\) may become very large or infinite. In that case,
small \(d_{\mathrm{obs},\sigma}\)-diameter alone does not guarantee small
belief diameter. FRR can still merge such directions, but only if the
decision diameter \(\Delta_Q\) is also small.
\end{remark}

\begin{lemma}[Observation-metric action-value Lipschitzness]
\label{lem:finite_obs_q_lipschitz}
Under Assumption~\ref{ass:obs_to_belief_stability},
\[
  \sup_{a\in\cA}|Q^*(b,a)-Q^*(b',a)|
  \le
  L_{Q,\mathrm{obs}}(\sigma)d_{\mathrm{obs},\sigma}(b,b'),
\]
where
\[
  L_{Q,\mathrm{obs}}(\sigma)
  :=
  L_Q^\Delta C_{\mathrm{obs}}(\sigma)
  =
  \frac{R_{\max}}{1-\gamma}C_{\mathrm{obs}}(\sigma).
\]
\end{lemma}

\begin{proof}
By Theorem~\ref{thm:finite_q_lipschitz}, for all $b,b'\in\DeltaS$,
\[
  \sup_{a\in\cA}|Q^*(b,a)-Q^*(b',a)|
  \le
  L_Q^\Delta\norm{b-b'}_1.
\]
Assumption~\ref{ass:obs_to_belief_stability} gives
\[
  \norm{b-b'}_1
  \le
  C_{\mathrm{obs}}(\sigma)d_{\mathrm{obs},\sigma}(b,b').
\]
Combining the two inequalities yields
\[
  \sup_{a\in\cA}|Q^*(b,a)-Q^*(b',a)|
  \le
  L_Q^\Delta C_{\mathrm{obs}}(\sigma)d_{\mathrm{obs},\sigma}(b,b'),
\]
which is the claimed bound.
\end{proof}

\begin{corollary}[Observation-metric FRR covering]
\label{cor:finite_obs_cover}
Suppose Assumption~\ref{ass:obs_to_belief_stability} holds. Let
$\{C_i\}_{i=1}^N$ be any covering of $\DeltaS$ satisfying
\[
  \diam_{\mathrm{obs},\sigma}(C_i)
  :=
  \sup_{b,b'\in C_i}d_{\mathrm{obs},\sigma}(b,b')
  \le
  \frac{\varepsilon}{L_{Q,\mathrm{obs}}(\sigma)}
\]
for every cell $C_i$. Then each cell has decision diameter at most
$\varepsilon$:
\[
  \Delta_Q(C_i)
  :=
  \sup_{b,b'\in C_i}\sup_{a\in\cA}
  |Q^*(b,a)-Q^*(b',a)|
  \le
  \varepsilon .
\]
Consequently, if
$\mathcal N_{\mathrm{diam}}(\DeltaS,d_{\mathrm{obs},\sigma},r)$ denotes the
minimum number of sets of $d_{\mathrm{obs},\sigma}$-diameter at most $r$
needed to cover $\DeltaS$, then
\[
  N_{\mathrm{obs},\sigma}(\varepsilon)
  \le
  \mathcal N_{\mathrm{diam}}
  \left(
    \DeltaS,
    d_{\mathrm{obs},\sigma},
    \frac{\varepsilon}{L_{Q,\mathrm{obs}}(\sigma)}
  \right).
\]
\end{corollary}

\begin{proof}
Let $C_i$ be one cell in the covering. For any $b,b'\in C_i$, Lemma~\ref{lem:finite_obs_q_lipschitz} gives
\[
  \sup_{a\in\cA}|Q^*(b,a)-Q^*(b',a)|
  \le
  L_{Q,\mathrm{obs}}(\sigma)d_{\mathrm{obs},\sigma}(b,b').
\]
Since $b,b'\in C_i$,
\[
  d_{\mathrm{obs},\sigma}(b,b')
  \le
  \diam_{\mathrm{obs},\sigma}(C_i)
  \le
  \frac{\varepsilon}{L_{Q,\mathrm{obs}}(\sigma)}.
\]
Therefore
\[
  \sup_{a\in\cA}|Q^*(b,a)-Q^*(b',a)|
  \le
  \varepsilon .
\]
Taking the supremum over all $b,b'\in C_i$ gives
$\Delta_Q(C_i)\le\varepsilon$. Since the same argument applies to every
cell, the covering is an $\varepsilon$-FRR covering. The covering-number
statement follows directly from the definition of
$\mathcal N_{\mathrm{diam}}$ with radius parameter
$r=\varepsilon/L_{Q,\mathrm{obs}}(\sigma)$.
\end{proof}

The observation-metric construction gives a route to smaller FRR covers when
the observable image of the simplex is lower-dimensional than the simplex
itself. In the finite POMDP, the predictive observation distribution for
action $a$ is the image of the predicted belief under the linear map
\[
  b\mapsto O_{a,\sigma}^\top P_a^\top b .
\]
Thus $d_{\mathrm{obs},\sigma}$ measures distances after projecting the
simplex through the family of observation-prediction maps
$\{O_{a,\sigma}^\top P_a^\top\}_{a\in\cA}$. If these maps have low rank or low
numerical rank on the relevant region of the simplex, then the observable
image may admit a lower-dimensional cover. However, because
$d_{\mathrm{obs},\sigma}$ is only a pseudometric, this dimension reduction is
valid for FRR only when an observation-to-value certification such as
Assumption~\ref{ass:obs_to_belief_stability} and
Lemma~\ref{lem:finite_obs_q_lipschitz} is available.

%============================================================================
\section{Finite Reliability Representations}
\label{sec:frr}
%============================================================================
\subsection{The correct smoothness object: the belief kernel}

The relevant smoothness object for dynamic programming is not the
fixed-observation Bayesian map \(\Phi_\sigma(b,u,y)\), but the controlled
belief-transition kernel \(\cK_\sigma(\cdot\mid b,u)\). The two answer
different questions: \(\Phi_\sigma(b,u,y)\) describes the posterior belief
\emph{after} a particular observation \(y\) has already occurred, taking
\(y\) as a fixed, known input. In contrast, \(\cK_\sigma(\cdot\mid b,u)\)
describes the distribution of the next belief \emph{before} the
observation is known: since the next belief depends on which \(y\) is
eventually realized, and that \(y\) is still random at the time control
\(u\) is chosen, the next belief is itself a random variable, and
\(\cK_\sigma\) is its law. It is this distribution over possible posteriors,
rather than any single realized posterior, that the Bellman operator must
integrate against when computing expected future value.

Because of this distinction, the following assumption should be read as a
regularity condition on the \emph{nominal planning model} --- that is, on
\(\cK_\sigma\) as used internally by the planner --- rather than as a claim
about the physical world. In particular, it is not a claim that the
physical system is exactly known, nor that Bayesian filtering is globally
contractive on all beliefs. The assumption is weaker and more local:
restricted to the reachable belief set used for planning, it requires only
that small perturbations of the current belief induce controlled
perturbations of the \emph{distribution} of the next belief, not that any
particular posterior update is contractive.

\begin{assumption}[Reachable belief-kernel invariance and Lipschitz modulus]
\label{ass:kernel_lipschitz}
The reachable belief set $\cB_{\mathrm{reach}}$ is invariant under the
controlled belief-transition kernel: $\cK_\sigma(\cB_{\mathrm{reach}}\mid
b,u) = 1$ for all $b\in\cB_{\mathrm{reach}}$, $u\in\cU$. Moreover, there
exists a finite constant $\beta_\sigma \ge 0$ such that, for all
$b,b'\in\cB_{\mathrm{reach}}$ and all $u\in\cU$,
\[
  W_1^{\cB}\big(\cK_\sigma(\cdot\mid b,u), \cK_\sigma(\cdot\mid b',u)\big)
  \le \beta_\sigma\, d_{W_1}(b,b'),
\]
where $W_1^{\cB}$ is the Wasserstein-1 distance on probability measures over
$\cB_{\mathrm{reach}}$, using $d_{W_1}$ as the ground metric.
\end{assumption}

\begin{remark}[Model-based interpretation] \label{rm:model-based-inter}
The invariance and Lipschitz conditions above are properties of the
\emph{belief model used for planning}, not of the true environment.
Since $\cK_\sigma$ is induced by the estimator the designer chooses ---
filter structure, noise model, and the set $\cB_{\mathrm{reach}}$ on which
it is analyzed --- both conditions can often be \emph{enforced by
construction} rather than verified empirically: $\cB_{\mathrm{reach}}$ can
be chosen forward-invariant under the nominal filter (e.g., a compact
Riccati-invariant set for a Kalman filter, or the full simplex
$\Delta(S)$ for a finite POMDP), and the filter can be designed or
certified to admit a finite modulus $\beta_\sigma$ on that set. No claim
is made that this nominal model exactly describes the true system;
model mismatch is a separate residual error, handled outside
$\beta_\sigma$. Thus $\beta_\sigma$ should be read as a certified modulus
of the \emph{planner's internal belief dynamics}, largely under the
designer's control, not as a universal constant of the physical world.
\end{remark}

\begin{remark}[Finite-POMDP justification]
In the finite-POMDP case, invariance of $\cB_{\mathrm{reach}}$ holds
automatically: taking $\cB_{\mathrm{reach}}=\DeltaS$, or any subset
positively invariant under $\{P_a\}_{a\in\cA}$, the Bayes update
$\tau_\sigma$ in~\eqref{eq:finite_bayes} always returns a point of the
simplex, so $\cK_\sigma(\cB_{\mathrm{reach}}\mid b,a)=1$ trivially. The
Lipschitz modulus $\beta_\sigma$ holds on any compact reachable subset of
the belief simplex on which the observation normalizers used by the
filter are uniformly bounded away from zero:
\[
  p_\sigma(z\mid b,a)\ge p_{\min}>0
\]
for all reachable $b$, actions $a$, and observations $z$ retained as
positive-probability branches of the model. On such a set, the Bayes
update is a smooth normalized linear map, and the finite mixture defining
$\cK_\sigma(\cdot\mid b,a)$ is Lipschitz in $b$.
\end{remark}

\begin{remark}[Linear-Gaussian justification]
For linear-Gaussian systems with Kalman filtering, Gaussian beliefs remain
Gaussian, so a compact Riccati reachable covariance set is, by
construction, forward-invariant under the filter recursion, giving
$\cK_\sigma(\cB_{\mathrm{reach}}\mid b,u)=1$ on that set. The belief state
is then finite-dimensional, and the belief-transition kernel is Gaussian
with parameters depending smoothly on the current belief parameters. On
compact subsets, this gives a finite Lipschitz modulus $\beta_\sigma$,
which can be computed or bounded from the system matrices and Kalman gain.
\end{remark}

\begin{remark}[Learned and approximate belief models]
Assumption~\ref{ass:kernel_lipschitz} can also be verified or estimated
for learned belief models, but unlike the two cases above, invariance of
$\cB_{\mathrm{reach}}$ is not automatic and must be checked or enforced
separately --- e.g., by verifying empirically that sampled belief rollouts
remain in $\cB_{\mathrm{reach}}$, or by projecting the learned update onto
$\cB_{\mathrm{reach}}$. For the Lipschitz modulus, if a dynamics model,
observation model, and belief encoder are implemented by Lipschitz neural
networks on compact domains, then the induced approximate belief update is
Lipschitz with a modulus bounded by the product of the component Lipschitz
constants. In practice, $\beta_\sigma$ may be obtained analytically from a
certified model class, or empirically from paired belief rollouts:
\[
  \widehat\beta_\sigma
  =
  \sup_{\substack{b\neq b'\\ u\in\cU}}
  \frac{
  W_1^{\cB}\!\left(
  \cK_\sigma(\cdot\mid b,u),
  \cK_\sigma(\cdot\mid b',u)
  \right)}
  {\dW(b,b')}.
\]
This makes FRR compatible with both model-based filters and learned
belief-state representations.
\end{remark}

Assumption~\ref{ass:kernel_lipschitz} requires invariance of
$\cB_{\mathrm{reach}}$ and a finite modulus $\beta_\sigma$; it does not
itself require Bayesian filtering to be contractive. The value-function
Lipschitz result below uses the stronger discounted condition
$\gamma\beta_\sigma<1$. This condition does not require the Bayesian
update itself to be contractive; it requires only that belief-kernel
sensitivity is small enough relative to the discount factor.

\subsection{Value and action-value Lipschitz continuity}

\begin{lemma}[Reward Lipschitzness on belief space]
\label{lem:reward_lipschitz}
Let
\[
  \rho(b,u):=\int_{\cX} r(x,u)\,b(dx)
\]
be the belief-averaged immediate reward. Under
Assumption~\ref{ass:regular}, for every \(u\in\cU\) and
\(b,b'\in\cB(\cX)\),
\[
  |\rho(b,u)-\rho(b',u)|
  \le
  L_r\,\dW(b,b').
\]
\end{lemma}

\begin{proof}
For fixed \(u\), Assumption~\ref{ass:regular} implies that
\(x\mapsto r(x,u)\) is \(L_r\)-Lipschitz. Hence, if \(L_r>0\),
\(r(\cdot,u)/L_r\) is an admissible 1-Lipschitz test function in the
definition of \(\dW\). Therefore,
\[
\begin{aligned}
  |\rho(b,u)-\rho(b',u)|
  &=
  \left|
  \int_{\cX} r(x,u)\,b(dx)
  -
  \int_{\cX} r(x,u)\,b'(dx)
  \right| \\
  &\le
  L_r\,\dW(b,b').
\end{aligned}
\]
If \(L_r=0\), then \(r(\cdot,u)\) is constant in \(x\), and the same bound
holds trivially.
\end{proof}

\begin{theorem}[Value-function Lipschitz bound]
\label{thm:value_lipschitz}
Suppose Assumptions~\ref{ass:regular} and~\ref{ass:kernel_lipschitz} hold
on $\cB_{\mathrm{reach}}$ and
\begin{equation}
\label{eqn:small_gain}
  \gamma\beta_\sigma<1.  
\end{equation}
Then $V^*$ is Lipschitz on $\cB_{\mathrm{reach}}$ with constant
\[
  L_V(\sigma)=\frac{L_r}{1-\gamma\beta_\sigma}.
\]
That is,
\[
  |V^*(b)-V^*(b')|
  \le
  L_V(\sigma)\,\dW(b,b')
  \qquad
  \forall b,b'\in\cB_{\mathrm{reach}}.
\]
\end{theorem}

\begin{remark}[$\gamma\beta_\sigma<1$ as a design target]
Both quantities in the condition $\gamma\beta_\sigma<1$ are, in a precise
sense, under the designer's control: $\gamma$ is fixed by the problem
formulation, and $\beta_\sigma$ is the modulus of the chosen estimator
(Assumption~\ref{ass:kernel_lipschitz}), which can often be reduced by
filter design (e.g., increasing observation informativeness or restricting
$\cB_{\mathrm{reach}}$ to a smaller invariant set). The condition is
therefore best read as a joint design target relating discounting and
estimator stability, not as an independent property of the physical
environment.
\end{remark}

\begin{proof}
Define the Bellman operator \(T\) by
\[
  (TV)(b)
  :=
  \sup_{u\in\cU}
  \left\{
  \rho(b,u)
  +
  \gamma
  \int_{\cB_{\mathrm{reach}}}V(b_+)\,
  \cK_\sigma(db_+\mid b,u)
  \right\}.
\]
Let \(V\) be \(L\)-Lipschitz on \(\cB_{\mathrm{reach}}\) with respect to
\(\dW\). For any \(b,b'\in\cB_{\mathrm{reach}}\),
\[
\begin{aligned}
  &|(TV)(b)-(TV)(b')| \\
  \le
  &\sup_{u\in\cU}
  \Bigg|
  \rho(b,u)-\rho(b',u) 
  +\gamma
  \int V\,d\cK_\sigma(\cdot\mid b,u) \\
  &\qquad \qquad -
  \gamma
  \int V\,d\cK_\sigma(\cdot\mid b',u)
  \Bigg| .
\end{aligned}
\]

By Lemma~\ref{lem:reward_lipschitz},
\[
  |\rho(b,u)-\rho(b',u)|
  \le
  L_r\,\dW(b,b').
\]
Since \(V\) is \(L\)-Lipschitz and \(W_1^{\cB}\) uses \(\dW\) as its ground
metric,
\[
\begin{aligned}
  &\left|
  \int V\,d\cK_\sigma(\cdot\mid b,u)
  -
  \int V\,d\cK_\sigma(\cdot\mid b',u)
  \right| \\
  &\quad \le
  L\,
  W_1^{\cB}\!\left(
  \cK_\sigma(\cdot\mid b,u),
  \cK_\sigma(\cdot\mid b',u)
  \right).  
\end{aligned}
\]
Assumption~\ref{ass:kernel_lipschitz} then gives
\[
  |(TV)(b)-(TV)(b')|
  \le
  \left(L_r+\gamma L\beta_\sigma\right)\dW(b,b').
\]
Thus \(T\) maps \(L\)-Lipschitz functions to
\((L_r+\gamma L\beta_\sigma)\)-Lipschitz functions.

Now choose
\[
  L=L_V(\sigma):=\frac{L_r}{1-\gamma\beta_\sigma}.
\]
Because \(\gamma\beta_\sigma<1\),
\[
  L_r+\gamma L_V(\sigma)\beta_\sigma
  =
  L_V(\sigma).
\]
Therefore \(T\) maps the class of \(L_V(\sigma)\)-Lipschitz functions into
itself.

Starting value iteration from any constant function \(V_0\), which is
\(L_V(\sigma)\)-Lipschitz, all iterates \(V_{k+1}=TV_k\) are
\(L_V(\sigma)\)-Lipschitz. Since \(T\) is a \(\gamma\)-contraction in
supremum norm for discounted rewards, \(V_k\) converges uniformly to the
unique fixed point \(V^*\). A uniform limit of functions with a common
Lipschitz constant has the same Lipschitz constant. Hence,
\[
  |V^*(b)-V^*(b')|
  \le
  L_V(\sigma)\,\dW(b,b')
\]
for all \(b,b'\in\cB_{\mathrm{reach}}\), as claimed.
\end{proof}

\begin{corollary}[Uniform Lipschitz bound on $Q^*$]
\label{cor:q_lipschitz}
Under the hypotheses of Theorem~\ref{thm:value_lipschitz},
\[
  |Q^*(b,u)-Q^*(b',u)|
  \le
  L_V(\sigma)\,\dW(b,b')
\]
for all $b,b'\in\cB_{\mathrm{reach}}$ and all $u\in\cU$.
\end{corollary}

\begin{proof}
Fix \(u\in\cU\). By definition~\eqref{eq:qstar}, \(Q^*(b,u)\) is exactly
the term inside the supremum defining \((TV^*)(b)\) in the proof of
Theorem~\ref{thm:value_lipschitz}, evaluated at this fixed \(u\). Since
\(V^*\) is \(L_V(\sigma)\)-Lipschitz (Theorem~\ref{thm:value_lipschitz}),
the same estimate used there --- before taking the supremum over \(u\) ---
gives, for this fixed \(u\),
\[
\begin{aligned}
  |Q^*(b,u)-Q^*(b',u)|
  &\le
  \bigl(L_r+\gamma L_V(\sigma)\beta_\sigma\bigr)\dW(b,b') \\
  &=
  L_V(\sigma)\dW(b,b'),
\end{aligned}
\]
where the equality is~\eqref{eqn:small_gain} applied as in
Theorem~\ref{thm:value_lipschitz}. Since \(u\in\cU\) was arbitrary, the
bound holds uniformly over \(u\in\cU\).
\end{proof}

The constant \(\beta_\sigma\) should be interpreted as a
belief-kernel stability modulus of the nominal planning model. The induced
constant \(L_V(\sigma)=L_r/(1-\gamma\beta_\sigma)\) is therefore not a
universal monotone function of sensor noise. In some systems, increasing
\(\sigma\) reduces observation distinguishability and permits coarser
reliable cells; in others, degraded observations can make the belief
dynamics less stable and increase \(\beta_\sigma\). The theory separates
these effects instead of imposing a universal monotonicity claim.

\subsection{FRR definition}

The previous subsection gives conditions under which \(Q^*(b,u)\) varies
in a Lipschitz-controlled manner over the reachable belief set. We now use
this smoothness to define belief-space cells whose internal variation is
small enough that a single representative belief can be used for
decision-making.

\begin{definition}[Decision diameter]
\label{def:decision_diameter}
For \(C\subseteq\cB_{\mathrm{reach}}\), define the decision diameter of
\(C\) by
\[
  \Delta_Q(C)
  :=
  \sup_{b,b'\in C}\sup_{u\in\cU}
  |Q^*(b,u)-Q^*(b',u)|.
\]
\end{definition}

The quantity \(\Delta_Q(C)\) measures the largest possible change in
action-value inside the cell \(C\), uniformly over admissible controls. A
small decision diameter therefore means that all beliefs in the cell are
nearly interchangeable for action selection.

\begin{definition}[Finite Reliability Representation]
\label{def:frr}
Fix the physical noise-floor parameter \(\sigma\)
(Section~\ref{sec:setup}) and a decision tolerance \(\varepsilon>0\), the
maximum decision diameter allowed within a cell. A finite collection
\(\{C_i,\hat b_i\}_{i=1}^N\) is a Finite Reliability Representation
(FRR) at level \((\sigma,\varepsilon)\) if
\[
  \cB_{\mathrm{reach}}\subseteq \bigcup_{i=1}^N C_i,
  \qquad
  \hat b_i\in C_i,
\]
and each cell satisfies
\[
  \Delta_Q(C_i)\le \varepsilon,
  \qquad
  i=1,\dots,N.
\]
The sets \(C_i\) are called reliability cells, and \(\hat b_i\) is the
representative belief for cell \(C_i\). The cells need not be disjoint:
an FRR is a cover of the reachable belief set, not a quotient space or an
equivalence-class partition.
\end{definition}

The following theorem gives a sufficient, metric-based way to construct an
FRR: any cover whose cells are small enough in Wasserstein diameter
automatically has small decision diameter.

\begin{theorem}[Constructive FRR cover]
\label{thm:frr_construction}
Suppose the hypotheses of Theorem~\ref{thm:value_lipschitz} hold. If each
cell satisfies
\[
  \diam_{\dW}(C_i)
  :=
  \sup_{b,b'\in C_i}\dW(b,b')
  \le
  \frac{\varepsilon}{L_V(\sigma)},
\]
then \(\{C_i,\hat b_i\}_{i=1}^N\) is an FRR at level
\((\sigma,\varepsilon)\).
\end{theorem}

\begin{proof}
By Corollary~\ref{cor:q_lipschitz}, for any \(b,b'\in C_i\) and any
\(u\in\cU\),
\[
  |Q^*(b,u)-Q^*(b',u)|
  \le
  L_V(\sigma)\dW(b,b').
\]
Taking the supremum over \(b,b'\in C_i\) and \(u\in\cU\) gives
\[
  \Delta_Q(C_i)
  \le
  L_V(\sigma)\diam_{\dW}(C_i).
\]
By hypothesis,
\[
  \diam_{\dW}(C_i)
  \le
  \frac{\varepsilon}{L_V(\sigma)},
\]
so \(\Delta_Q(C_i)\le\varepsilon\). Thus every cell \(C_i\) satisfies the
decision-diameter condition in Definition~\ref{def:frr}. Since the
collection covers \(\cB_{\mathrm{reach}}\) and each cell has a
representative \(\hat b_i\in C_i\), the collection is an FRR at level
\((\sigma,\varepsilon)\).
\end{proof}

\begin{corollary}[Action-noise-dependent FRR radius]
\label{cor:action_noise_frr_radius}
Let
\[
  \theta=(\sigma_y,\sigma_u)
\]
denote a combined physical noise-floor parameter, where \(\sigma_y\)
represents sensing uncertainty and \(\sigma_u\) represents action-execution
uncertainty. Write \(Q_\theta^*\) for the optimal action-value function of
the corresponding noisy POMDP model. Suppose that, on
\(\cB_{\mathrm{reach}}\), \(Q_\theta^*\) admits the certified Lipschitz
bound
\[
  |Q_\theta^*(b,u)-Q_\theta^*(b',u)|
  \le
  L_Q(\theta)\,\dW(b,b')
\]
for all \(b,b'\in\cB_{\mathrm{reach}}\) and \(u\in\cU\).
Then any cover satisfying
\[
  \diam_{\dW}(C_i)
  \le
  \frac{\varepsilon}{L_Q(\theta)}
\]
is an FRR at level \((\theta,\varepsilon)\).
\end{corollary}

\begin{proof}
The proof is identical to the proof of
Theorem~\ref{thm:frr_construction}, with \(L_V(\sigma)\) replaced by the
certified action-value Lipschitz constant \(L_Q(\theta)\). For any
\(b,b'\in C_i\) and any \(u\in\cU\),
\[
  |Q_\theta^*(b,u)-Q_\theta^*(b',u)|
  \le
  L_Q(\theta)\dW(b,b').
\]
Taking the supremum over \(b,b'\in C_i\) and \(u\in\cU\) yields
\[
  \Delta_{Q_\theta}(C_i)
  \le
  L_Q(\theta)\diam_{\dW}(C_i).
\]
By hypothesis,
\[
  \diam_{\dW}(C_i)
  \le
  \frac{\varepsilon}{L_Q(\theta)},
\]
so \(\Delta_{Q_\theta}(C_i)\le\varepsilon\), and \(C_i\) satisfies the FRR
decision-diameter condition at level \((\theta,\varepsilon)\).
\end{proof}

\begin{lemma}[Monotonicity transfer from kernel modulus to $Q^*$-sensitivity]
\label{lem:beta_to_LQ_monotone}
Suppose, for a family of noise parameters \(\theta=(\sigma_y,\sigma_u)\),
Assumption~\ref{ass:kernel_lipschitz} holds with modulus \(\beta_\theta\)
and \(\gamma\beta_\theta<1\), and that \(L_r\) is independent of \(\theta\)
(Assumption~\ref{ass:regular}). Then the certified Lipschitz constant
\[
  L_Q(\theta):=\frac{L_r}{1-\gamma\beta_\theta}
\]
is nondecreasing in \(\beta_\theta\) (strictly increasing when \(L_r>0\)).
In particular, fixing \(\sigma_y\), if
\[
  \beta_{\sigma_y,\sigma_u^{(2)}}
  \le
  \beta_{\sigma_y,\sigma_u^{(1)}},
\]
then
\[
  L_Q(\sigma_y,\sigma_u^{(2)})
  \le
  L_Q(\sigma_y,\sigma_u^{(1)}).
\]
\end{lemma}

\begin{proof}
The map \(\beta\mapsto L_r/(1-\gamma\beta)\) has derivative
\(\gamma L_r/(1-\gamma\beta)^2\ge 0\) on \([0,1/\gamma)\), hence is
nondecreasing on this interval, and strictly increasing when \(L_r>0\).
The displayed implication follows immediately by evaluating this map at
\(\beta_{\sigma_y,\sigma_u^{(1)})}\) and \(\beta_{\sigma_y,\sigma_u^{(2)}}\).
\end{proof}

Combining Lemma~\ref{lem:beta_to_LQ_monotone} with
Corollary~\ref{cor:action_noise_frr_radius} gives a comparative-statics
consequence stated directly in terms of the belief-kernel modulus
\(\beta_\theta\), rather than the derived quantity \(L_Q(\theta)\).

\begin{corollary}[Kernel-modulus comparison and FRR radius]
\label{cor:kernel_modulus_frr_radius}
Fix the sensing noise level \(\sigma_y\), and consider two
action-execution noise levels \(\sigma_u^{(1)}\) and \(\sigma_u^{(2)}\)
satisfying the hypotheses of Lemma~\ref{lem:beta_to_LQ_monotone}. If
\[
  \beta_{\sigma_y,\sigma_u^{(2)}}
  \le
  \beta_{\sigma_y,\sigma_u^{(1)}},
\]
then the admissible FRR cell diameter under \(\sigma_u^{(2)}\) is at least
as large as that under \(\sigma_u^{(1)}\):
\[
  \frac{\varepsilon/L_Q(\sigma_y,\sigma_u^{(2)})}
       {\varepsilon/L_Q(\sigma_y,\sigma_u^{(1)})}
  =
  \frac{L_Q(\sigma_y,\sigma_u^{(1)})}
       {L_Q(\sigma_y,\sigma_u^{(2)})}
  \ge 1.
\]
Thus a belief-kernel modulus that does not increase with action-execution
noise permits coarser FRR cells at the same decision tolerance.
\end{corollary}

\begin{proof}
By Lemma~\ref{lem:beta_to_LQ_monotone}, the hypothesis
\(\beta_{\sigma_y,\sigma_u^{(2)}}\le\beta_{\sigma_y,\sigma_u^{(1)}}\)
gives \(L_Q(\sigma_y,\sigma_u^{(2)})\le L_Q(\sigma_y,\sigma_u^{(1)})\).
Dividing both sides of this inequality by
\(L_Q(\sigma_y,\sigma_u^{(2)})>0\) gives
\[
  1
  \le
  \frac{L_Q(\sigma_y,\sigma_u^{(1)})}
       {L_Q(\sigma_y,\sigma_u^{(2)})},
\]
which, by Corollary~\ref{cor:action_noise_frr_radius}, is exactly the
ratio of admissible FRR cell diameters under \(\sigma_u^{(2)}\) versus
\(\sigma_u^{(1)}\).
\end{proof}

\begin{remark}[Scope of the kernel-modulus comparison]
\label{rem:action_noise_scope}
Corollary~\ref{cor:kernel_modulus_frr_radius} should not be read as a
universal monotonicity theorem. Increasing action-execution uncertainty
can make the control problem harder and can reduce the optimal value, and
the hypothesis \(\beta_{\sigma_y,\sigma_u^{(2)}}\le
\beta_{\sigma_y,\sigma_u^{(1)}}\) is not implied merely by
Theorem~\ref{thm:finite_actuator_data_processing}: that theorem shows
actuator degradation contracts \emph{command} distinguishability
\(d_{\mathrm{exec},\sigma_u}\), whereas \(\beta_\theta\) measures the
sensitivity of the belief-transition kernel to perturbations of the
current belief. These are related but distinct objects, and the transfer
from one to the other is model-dependent.

As with the discounted sensitivity condition \(\gamma\beta_\sigma<1\)
(Section~\ref{sec:frr}), the hypothesis
\(\beta_{\sigma_y,\sigma_u^{(2)}}\le\beta_{\sigma_y,\sigma_u^{(1)}}\) is a
property of the \emph{designed estimator} evaluated at two operating
points, not a fact about the physical actuator. Since \(\beta_\theta\) is
induced by the filter structure and noise model the designer chooses
(Remark~\ref{rm:model-based-inter}), whether it is nonincreasing
in \(\sigma_u\) depends on that design: a fixed filter re-analyzed at
higher actuation noise may see \(\beta_\theta\) increase rather than
decrease, while a filter re-tuned for the higher-noise regime (e.g., with
adjusted gain or process-noise covariance) may be designed specifically to
keep \(\beta_\theta\) from growing. The corollary states only that,
\emph{if} the estimator design achieves this, then larger belief cells can
be certified at the same decision tolerance. This is the mechanism
observed in the finite-POMDP numerical example: as commanded actions
become less distinguishable in execution, the action-value variation
decreases and fewer reliability cells are needed.
\end{remark}

\subsection{Policy sufficiency}

The previous construction shows how to obtain cells with small decision
diameter. The next result shows why decision diameter is the correct
quantity: once every cell has small \(Q^*\)-variation uniformly over
actions, a policy that chooses actions only from cell representatives is
near-optimal.

\begin{theorem}[Policy sufficiency]
\label{thm:policy_sufficiency}
Suppose Assumption~\ref{ass:maximizers} holds, and let
\(\{C_i,\hat b_i\}_{i=1}^N\) be any FRR at level
\((\sigma,\varepsilon)\). Define the cell-index map
\[
\begin{aligned}
    i&:\cB_{\mathrm{reach}}\to\{1,\dots,N\}, \\
  i(b)&:=\min\{j\in\{1,\dots,N\}: b\in C_j\},
\end{aligned}
\]
which is well-defined since \(\{C_i\}_{i=1}^N\) covers
\(\cB_{\mathrm{reach}}\); when \(b\) lies in more than one cell, \(i(b)\)
selects the smallest such index. Define the cell-constant policy
\[
  \bar\pi(b)
  \in
  \arg\max_{u\in\cU} Q^*(\hat b_{i(b)},u).
\]
Then
\[
  \norm{V^*-V^{\bar\pi}}_\infty
  \le
  \frac{2\varepsilon}{1-\gamma}.
\]
\end{theorem}

\begin{proof}
Fix \(b\in\cB_{\mathrm{reach}}\), and let \(i=i(b)\). Since
\(b,\hat b_i\in C_i\) and \(C_i\) has decision diameter at most
\(\varepsilon\),
\[
  |Q^*(b,u)-Q^*(\hat b_i,u)|
  \le
  \varepsilon
  \qquad
  \forall u\in\cU .
\]
Let
\[
  u^*(b)\in\arg\max_{u\in\cU} Q^*(b,u)
\]
be an optimal action at \(b\). By the definition of \(\bar\pi\),
\[
  Q^*(\hat b_i,\bar\pi(b))
  \ge
  Q^*(\hat b_i,u^*(b)).
\]
Combining the decision-diameter bound (applied once at \(u=\bar\pi(b)\)
and once at \(u=u^*(b)\)) with the defining optimality property of
\(\bar\pi\) at \(\hat b_i\) therefore gives:
\[
\begin{aligned}
  &Q^*(b,\bar\pi(b)) \\
  &\quad\ge
  Q^*(\hat b_i,\bar\pi(b))-\varepsilon
  &&\text{(diameter bound, \(u=\bar\pi(b)\))}\\
  &\quad\ge
  Q^*(\hat b_i,u^*(b))-\varepsilon
  &&\text{(optimality of \(\bar\pi(b)\) at \(\hat b_i\))}\\
  &\quad\ge
  Q^*(b,u^*(b))-2\varepsilon
  &&\text{(diameter bound, \(u=u^*(b)\))}\\
  &\quad=
  V^*(b)-2\varepsilon .
  &&\text{(definition of \(V^*\))}
\end{aligned}
\]
Thus \(\bar\pi\) is \(2\varepsilon\)-greedy with respect to \(Q^*\):
\[
  V^*(b)-Q^*(b,\bar\pi(b))\le 2\varepsilon
  \qquad
  \forall b\in\cB_{\mathrm{reach}}.
\]

Now compare \(V^*\) and \(V^{\bar\pi}\). For any \(b\),
\[
\begin{aligned}
  &V^*(b)-V^{\bar\pi}(b) \\
  &=
  V^*(b)-Q^*(b,\bar\pi(b)) +
  Q^*(b,\bar\pi(b))-V^{\bar\pi}(b) \\
  &\le
  2\varepsilon
  +
  \gamma
  \int_{\cB_{\mathrm{reach}}}
  \bigl(V^*(b_+)-V^{\bar\pi}(b_+)\bigr)
  \cK_\sigma(db_+\mid b,\bar\pi(b)) \\
  &\le
  2\varepsilon
  +
  \gamma\norm{V^*-V^{\bar\pi}}_\infty .
\end{aligned}
\]
Taking the supremum over \(b\in\cB_{\mathrm{reach}}\) gives
\[
  \norm{V^*-V^{\bar\pi}}_\infty
  \le
  2\varepsilon
  +
  \gamma\norm{V^*-V^{\bar\pi}}_\infty .
\]
Rearranging yields
\[
  \norm{V^*-V^{\bar\pi}}_\infty
  \le
  \frac{2\varepsilon}{1-\gamma}.
\]
\end{proof}

Theorem~\ref{thm:policy_sufficiency} is the main robustness property of
the framework. Its proof uses only the decision-diameter bound
\(\Delta_Q(C_i)\le\varepsilon\); it never refers to how a cover achieving
this bound was constructed. Consequently, the \(2\varepsilon/(1-\gamma)\)
guarantee transfers unchanged to any FRR cover, whether obtained from an
analytic Lipschitz bound (Theorem~\ref{thm:frr_construction}), a
linear-Gaussian approximation, a finite-POMDP simplex cover
(Section~\ref{sec:finite_pomdp}), a particle-filter estimate, or direct
empirical certification of \(Q^*\)-variation
(Section~\ref{subsec:numerical_certification}).

%============================================================================
\section{Reliability Entropy and Usable Policy Capacity}
\label{sec:entropy}
%============================================================================

The policy-sufficiency theorem shows that an FRR cover is enough to obtain
a certified near-optimal cell-constant policy. We now ask how many cells are
needed. This gives a representation-complexity measure for belief-space
decision-making under a specified physical noise floor and decision
tolerance. This question is also the natural point of contact between FRR
and learned belief representations: if a compact belief space can be
covered efficiently, an even more efficient cover may be obtainable by
first mapping beliefs into a lower-dimensional learned latent space and
covering there instead. Section~\ref{subsec:encoder_capacity} makes this
precise, giving both a constructive FRR cover built from an encoder-critic
pair and an explicit condition under which the resulting reliability
entropy is provably smaller than that of the raw belief space.

\begin{definition}[Reliability entropy]
\label{def:entropy}
Let \(N_{\mathrm{FRR}}(\sigma,\varepsilon)\) be the smallest number of cells
in an FRR cover of \(\cB_{\mathrm{reach}}\) at level
\((\sigma,\varepsilon)\). Define
\[
  \cH_{\mathrm{FRR}}(\sigma,\varepsilon)
  :=
  \log N_{\mathrm{FRR}}(\sigma,\varepsilon).
\]
\end{definition}

This quantity is closer to a rate-distortion function or metric entropy
than to Shannon entropy. It counts the logarithm of the minimum number of
decision-relevant belief messages that must be preserved at tolerance
\(\varepsilon\) under the physical noise-floor parameter \(\sigma\).

\begin{proposition}[Metric-entropy upper bound]
\label{prop:entropy_bound}
Suppose the hypotheses of Theorem~\ref{thm:value_lipschitz} hold. Let
\(N_{\cB}(r)\) denote the minimum number of \(\dW\)-balls of radius \(r\)
needed to cover \(\cB_{\mathrm{reach}}\). Then
\[
  N_{\mathrm{FRR}}(\sigma,\varepsilon)
  \le
  N_{\cB}\!\left(\frac{\varepsilon}{2L_V(\sigma)}\right).
\]
Suppose, in addition, that \(\cB_{\mathrm{reach}}\) has diameter
\[
  D_B:=\sup_{b,b'\in\cB_{\mathrm{reach}}}\dW(b,b')
\]
and effective covering dimension \(q\), in the sense that
\(N_{\cB}(r)\lesssim(D_B/r)^q\) as \(r\to 0\), with a constant depending
only on \(q\) and the geometry of \(\cB_{\mathrm{reach}}\), not on \(r\).
Here \(\lesssim\) hides this constant, since only the exponent \(q\) and
the scaling in \(1/\varepsilon\) are needed below, not its exact value.
Then
\[
  N_{\mathrm{FRR}}(\sigma,\varepsilon)
  \lesssim
  \left(
  \frac{2L_V(\sigma)D_B}{\varepsilon}
  \right)^q.
\]
\end{proposition}

\begin{proof}
Let
\[
  r=\frac{\varepsilon}{2L_V(\sigma)}.
\]
By definition of \(N_{\cB}(r)\), the reachable belief set can be covered by
\(N_{\cB}(r)\) many \(\dW\)-balls of radius \(r\). Each such ball has
\(\dW\)-diameter at most \(2r\), hence
\[
  \diam_{\dW}(C_i)\le 2r
  =
  \frac{\varepsilon}{L_V(\sigma)}.
\]
By Theorem~\ref{thm:frr_construction}, each ball is therefore a valid FRR
cell. Hence
\[
  N_{\mathrm{FRR}}(\sigma,\varepsilon)
  \le
  N_{\cB}\!\left(\frac{\varepsilon}{2L_V(\sigma)}\right).
\]
For the second claim, substitute \(r=\varepsilon/(2L_V(\sigma))\) into the
covering-dimension scaling hypothesis \(N_{\cB}(r)\lesssim(D_B/r)^q\)
stated above, which gives
\[
\begin{aligned}
  N_{\mathrm{FRR}}(\sigma,\varepsilon)
  &\le
  N_{\cB}\!\left(\frac{\varepsilon}{2L_V(\sigma)}\right)
  \lesssim
  \left(
  \frac{D_B}{\varepsilon/(2L_V(\sigma))}
  \right)^q \\
  & \quad = 
  \left(
  \frac{2L_V(\sigma)D_B}{\varepsilon}
  \right)^q.  
\end{aligned}
\]
\end{proof}

\begin{remark}[Scope of the covering-dimension hypothesis]
\label{rem:covering_dimension_scope}
Effective covering dimension \(q\) is a property of the metric geometry
of \(\cB_{\mathrm{reach}}\), not a requirement that the belief space be
literally finite-dimensional. It holds automatically for finite POMDPs,
with \(q=n-1\) (Corollary~\ref{cor:finite_simplex_cover}), and for
belief families confined to a finite-dimensional parametric class, such
as Gaussian beliefs with covariance restricted to a compact Riccati
reachable set. For more general continuous belief spaces, finite
covering dimension is a genuine additional hypothesis that must be
verified for the specific belief-transition model at hand; it need not
hold for an unrestricted compact set of probability measures under
\(\dW\).
\end{remark}

\begin{remark}[What can and cannot be claimed about monotonicity]
A noisier sensor decreases the distinguishability of predictive observations
by Proposition~\ref{prop:data_processing}. However, this fact alone does not
imply that \(N_{\mathrm{FRR}}(\sigma,\varepsilon)\) must decrease for every
closed-loop POMDP. The FRR cell count is governed by decision diameter, which
depends on information acquisition, executable action resolution, and
belief-kernel stability. Sensor noise may reduce observation
distinguishability while also making the belief dynamics less stable or the
value function more sensitive to belief errors. A decreasing reliability
entropy is therefore a conditional result. It is expected in regimes where the
physical noise floor collapses distinctions that are both observationally and
executably irrelevant, for example when actuation uncertainty makes distinct
commanded actions nearly indistinguishable in their effects. It should not be
claimed from sensor degradation alone. When sensing and actuation noise
increase together, a decreasing reliability entropy is more plausible
whenever the actuation-driven collapse of command distinguishability
(Corollary~\ref{cor:kernel_modulus_frr_radius}) dominates any destabilizing
effect of the sensor component on \(\beta_\theta\); this was the pattern
observed in the coupled finite-POMDP experiment
(Section~\ref{subsec:finite_pomdp_numerics}), though the coupled
particle-filter diagnostics (Section~\ref{subsec:particle_filter_numerics})
show this need not hold once sampled certification is also sensitive to
reachable-set geometry and critic approximation. This should still be read
as a regime-dependent tendency, not a monotonicity guarantee.
\end{remark}

\begin{corollary}[Usable capacity of cell-constant policies]
\label{cor:capacity}
Let \(\{C_i,\hat b_i\}_{i=1}^N\) be an FRR cover at level
\((\sigma,\varepsilon)\), and let \(\Pi_{\{C_i\}}\) be the class of
policies that are constant on the cells of this cover. Then every policy
in \(\Pi_{\{C_i\}}\) induces at most \(N\) certified belief regions, or
equivalently at most \(\log N\) nats of certified cell identity.
In particular, a minimum-cardinality FRR cover supports
\[
  \log N_{\mathrm{FRR}}(\sigma,\varepsilon)
  =
  \cH_{\mathrm{FRR}}(\sigma,\varepsilon)
\]
nats of certified reliability-cell identity. Policy architectures with
substantially larger effective capacity may still be useful for
optimization or function approximation, but any additional decision
boundaries are not certified as physically reliable at tolerance
\(\varepsilon\) unless the refined cells are themselves shown to have
decision diameter at most \(\varepsilon\).
\end{corollary}

\begin{proof}
A policy \(\pi\in\Pi_{\{C_i\}}\) is constant on each \(C_i\), so \(\pi\)
takes at most one value per cell; since there are \(N\) cells, \(\pi\)
can assign at most \(N\) distinct actions in total, indexed by the cell
label \(i\in\{1,\dots,N\}\). Thus the cover induces at most \(N\)
certified belief regions, corresponding to at most \(\log N\) nats of
cell identity. If the cover has minimum cardinality
\(N_{\mathrm{FRR}}(\sigma,\varepsilon)\), this number becomes
\[
  \log N_{\mathrm{FRR}}(\sigma,\varepsilon)
  =
  \cH_{\mathrm{FRR}}(\sigma,\varepsilon).
\]
Any additional decision boundary introduced inside a cell is not certified
by the FRR guarantee unless the resulting refined cells also satisfy the
decision-diameter condition in Definition~\ref{def:frr}.
\end{proof}

\subsection{Encoder-induced reliability covers}
\label{subsec:encoder_capacity}

Reliability entropy governs a different stage of the problem than the VC
dimension of the encoder-critic architecture, and the two should not be
conflated. VC dimension (or, more generally, the sample complexity) of
the parametric class from which \(E_\psi\) and \(q_\theta\) are drawn
governs \emph{how much training data is needed to learn} an
encoder-critic pair achieving a target approximation error
\(\eta_{\mathrm{enc}}\): a larger, higher-capacity architecture can
typically reach a smaller \(\eta_{\mathrm{enc}}\) for fixed data, or
requires more data to reach the same \(\eta_{\mathrm{enc}}\), by the
usual approximation-versus-estimation tradeoff. This is a statement about
\emph{learning} the representation.

\(\cH_{\mathrm{FRR}}(\sigma,\varepsilon)\), by contrast, governs
\emph{using} the representation once trained: given an already-learned
encoder-critic pair with achieved error \(\eta_{\mathrm{enc}}\) and latent
dimension \(d\), it measures the number of nats needed to certify a
decision-relevant belief message at tolerance \(\varepsilon\)
(Definition~\ref{def:entropy}). This quantity depends only on the
\emph{achieved} \(\eta_{\mathrm{enc}}\) and on \(d\)
(Theorem~\ref{thm:encoder_frr_construction} below), not on the VC
dimension or parameter count of the architecture that produced them: two
encoders of very different size and training cost, if they happen to
achieve the same \(\eta_{\mathrm{enc}}\) and \(d\), induce the same
certified reliability entropy.

The two notions therefore sit on opposite sides of the training/deployment
boundary. VC dimension of \((E_\psi,q_\theta)\) is what determines
whether a target \(\eta_{\mathrm{enc}}\) is learnable at all from a given
amount of data, and in particular whether \(\eta_{\mathrm{enc}}
<\varepsilon/2\) can be achieved
(Proposition~\ref{prop:encoder_floor} below) --- a training-time,
sample-complexity question. \(\cH_{\mathrm{FRR}}\) is what determines how
much certified decision capacity that trained encoder then supports at
deployment --- a rate-distortion question about the representation's
output, not its architecture. A larger encoder can therefore be
\emph{necessary} to reach the approximation floor required for
certification, while contributing nothing to \(\cH_{\mathrm{FRR}}\)
beyond what \(\eta_{\mathrm{enc}}\) and \(d\) already capture; conversely,
among encoders that all clear the \(\eta_{\mathrm{enc}}<\varepsilon/2\)
threshold, the one with the smallest latent dimension \(d\) is preferred,
regardless of how each was trained.

Suppose a learned encoder
\[
  E_\psi:\cB_{\mathrm{reach}}\to Z
\]
maps reachable beliefs into a latent space \(Z\subset\R^d\), and suppose a
latent critic \(q_\theta:Z\times\cU\to\R\) approximates the optimal
action-value function:
\[
  Q^*(b,u)\approx q_\theta(E_\psi(b),u).
\]
Define the uniform encoder-critic approximation error
\[
  \eta_{\mathrm{enc}}
  :=
  \sup_{b\in\cB_{\mathrm{reach}}}\sup_{u\in\cU}
  \left|
  Q^*(b,u)-q_\theta(E_\psi(b),u)
  \right|,
\]
and assume \(q_\theta\) is \(L_{Q,z}\)-Lipschitz in the latent variable:
\[
  |q_\theta(z,u)-q_\theta(z',u)|
  \le
  L_{Q,z}\norm{z-z'}
  \qquad
  \forall z,z'\in Z,\; u\in\cU .
\]

\begin{lemma}[Encoder-induced decision-diameter bound]
\label{lem:encoder_decision_diameter}
For any $C\subseteq\cB_{\mathrm{reach}}$,
\[
  \Delta_Q(C)
  \le
  2\eta_{\mathrm{enc}}
  +
  L_{Q,z}\,\diam\bigl(E_\psi(C)\bigr).
\]
\end{lemma}

\begin{proof}
For any $b,b'\in C$ and $u\in\cU$, the triangle inequality and the
defining properties of $\eta_{\mathrm{enc}}$ and $L_{Q,z}$ give
\[
\begin{aligned}
  &|Q^*(b,u)-Q^*(b',u)| \\
  &\le
  |Q^*(b,u)-q_\theta(E_\psi(b),u)| \\
  &\quad+
  |q_\theta(E_\psi(b),u)-q_\theta(E_\psi(b'),u)| \\
  &\quad+
  |q_\theta(E_\psi(b'),u)-Q^*(b',u)| \\
  &\le
  \eta_{\mathrm{enc}}
  +
  L_{Q,z}\norm{E_\psi(b)-E_\psi(b')}
  +
  \eta_{\mathrm{enc}} .
\end{aligned}
\]
Taking the supremum over $b,b'\in C$ and $u\in\cU$ gives
\[
  \Delta_Q(C)
  \le
  2\eta_{\mathrm{enc}}
  +
  L_{Q,z}\,\diam\bigl(E_\psi(C)\bigr).
\]
\end{proof}

Lemma~\ref{lem:encoder_decision_diameter} decomposes the certification
error into two terms of different character. The first, $2\eta_{\mathrm{enc}}$,
is an approximation floor caused by the encoder and critic; it cannot be
reduced by refining the latent cover unless the learned representation
itself improves. The second, $L_{Q,z}\,\diam(E_\psi(C))$, is a covering
term that can be made arbitrarily small by shrinking cells in latent
space. This split is what makes learned encoders certifiable rather than
merely heuristic: it isolates exactly which part of the error is
representation-limited and which part is coverable.

\begin{theorem}[Learned FRR cover via encoder compression]
\label{thm:encoder_frr_construction}
Suppose $\varepsilon>2\eta_{\mathrm{enc}}$, and let
$\{B_Z(z_j,r_Z)\}_{j=1}^{M}$ be a cover of $E_\psi(\cB_{\mathrm{reach}})$
by balls in $Z$ of radius
\[
  r_Z
  =
  \frac{\varepsilon-2\eta_{\mathrm{enc}}}{2L_{Q,z}}.
\]
Define $C_j:=E_\psi^{-1}\bigl(B_Z(z_j,r_Z)\bigr)\cap\cB_{\mathrm{reach}}$
and choose any $\hat b_j\in C_j$ for each nonempty $C_j$. Then the
resulting collection $\{C_j,\hat b_j\}$ is an FRR at level
$(\sigma,\varepsilon)$ with at most $M$ cells. Consequently,
\[
  N_{\mathrm{FRR}}(\sigma,\varepsilon)
  \le
  N_Z(r_Z),
\]
where $N_Z(r)$ denotes the covering number of $E_\psi(\cB_{\mathrm{reach}})$
by $Z$-balls of radius $r$. If $Z$ has effective covering dimension $d$
and diameter $D_Z$, then
\[
  N_{\mathrm{FRR}}(\sigma,\varepsilon)
  \lesssim
  \left(
  \frac{2L_{Q,z}D_Z}
       {\varepsilon-2\eta_{\mathrm{enc}}}
  \right)^{d}.
\]
\end{theorem}

\begin{proof}
The sets $\{B_Z(z_j,r_Z)\}_{j=1}^M$ cover $E_\psi(\cB_{\mathrm{reach}})$
by construction, so $\{C_j\}_{j=1}^M$ covers $\cB_{\mathrm{reach}}$. For
each nonempty $C_j$, $\diam(E_\psi(C_j))\le 2r_Z$, since $C_j$ is the
preimage of a ball of radius $r_Z$. By
Lemma~\ref{lem:encoder_decision_diameter},
\[
  \Delta_Q(C_j)
  \le
  2\eta_{\mathrm{enc}}+L_{Q,z}\cdot 2r_Z
  =
  2\eta_{\mathrm{enc}}
  +
  (\varepsilon-2\eta_{\mathrm{enc}})
  =
  \varepsilon .
\]
Thus every nonempty $C_j$ satisfies the decision-diameter condition in
Definition~\ref{def:frr}, and the collection of nonempty cells (at most
$M$ of them) is an FRR at level $(\sigma,\varepsilon)$. The covering-number
statement follows by definition of $N_Z(r_Z)$, and the asymptotic bound
follows from the standard scaling $N_Z(r)\lesssim(D_Z/r)^d$ for a set of
covering dimension $d$ and diameter $D_Z$, evaluated at $r=r_Z$.
\end{proof}

\begin{corollary}[When encoder compression reduces certified capacity]
\label{cor:encoder_capacity_comparison}
Suppose the hypotheses of Proposition~\ref{prop:entropy_bound} and
Theorem~\ref{thm:encoder_frr_construction} both hold, so that
\[
  N_{\mathrm{FRR}}(\sigma,\varepsilon)
  \lesssim
  \min\left\{
  \left(\frac{2L_V(\sigma)D_B}{\varepsilon}\right)^{q}
  ,\;
  \left(\frac{2L_{Q,z}D_Z}{\varepsilon-2\eta_{\mathrm{enc}}}\right)^{d}
  \right\},
\]
where $q$ is the covering dimension of $\cB_{\mathrm{reach}}$ and $d$ is
the covering dimension of $Z$. If $d<q$ and $\eta_{\mathrm{enc}}$ is
small enough that
\[
  \frac{2L_{Q,z}D_Z}{\varepsilon-2\eta_{\mathrm{enc}}}
  \le
  \frac{2L_V(\sigma)D_B}{\varepsilon}
\]
holds with room to spare, the encoder-based bound eventually dominates
as $\varepsilon\to 0$, and the certified reliability entropy satisfies
\[
  \cH_{\mathrm{FRR}}(\sigma,\varepsilon)
  =
  O\bigl(d\log(1/(\varepsilon-2\eta_{\mathrm{enc}}))\bigr)
\]
rather than the raw-belief-space rate
$O(q\log(1/\varepsilon))$. Thus a learned encoder strictly reduces usable
policy capacity, in the rate-distortion sense of
Definition~\ref{def:entropy}, precisely when its latent dimension is lower
than the belief space's covering dimension and its approximation floor
$2\eta_{\mathrm{enc}}$ is small relative to the target tolerance
$\varepsilon$.
\end{corollary}

\begin{proof}
Immediate from comparing the exponents and prefactors in the two bounds of
Proposition~\ref{prop:entropy_bound} and
Theorem~\ref{thm:encoder_frr_construction}: for fixed constants, a smaller
exponent $d<q$ dominates the polynomial rate in $1/\varepsilon$ as
$\varepsilon\to0$ (equivalently, as $\varepsilon-2\eta_{\mathrm{enc}}\to0$
from above), provided $\eta_{\mathrm{enc}}$ remains fixed and strictly
below $\varepsilon/2$ throughout. Taking logarithms of the smaller bound
gives the stated rate for $\cH_{\mathrm{FRR}}$.
\end{proof}

\begin{proposition}[Necessity of the encoder tolerance margin]
\label{prop:encoder_floor}
If $\eta_{\mathrm{enc}}\ge\varepsilon/2$, then the bound of
Lemma~\ref{lem:encoder_decision_diameter} cannot certify any cell at
decision tolerance $\varepsilon$, regardless of how finely
$E_\psi(\cB_{\mathrm{reach}})$ is covered: for every nonempty
$C\subseteq\cB_{\mathrm{reach}}$,
\[
  2\eta_{\mathrm{enc}}
  +
  L_{Q,z}\,\diam\bigl(E_\psi(C)\bigr)
  \ge
  2\eta_{\mathrm{enc}}
  \ge
  \varepsilon ,
\]
even in the degenerate limit $\diam(E_\psi(C))\to 0$.
\end{proposition}

\begin{proof}
Immediate, since $L_{Q,z}\,\diam(E_\psi(C))\ge 0$ and
$2\eta_{\mathrm{enc}}\ge\varepsilon$ by hypothesis.
\end{proof}

\begin{remark}[Scope of Proposition~\ref{prop:encoder_floor}]
Proposition~\ref{prop:encoder_floor} is a statement about the
encoder-critic certification route of
Lemma~\ref{lem:encoder_decision_diameter}, not about $\Delta_Q(C)$
itself. It shows that once the encoder-critic pair's approximation error
reaches $\varepsilon/2$, no amount of latent-space refinement can produce
a certificate at tolerance $\varepsilon$ through this construction --- the
approximation floor, not the covering term, is the binding constraint.
This does not preclude the existence of some other FRR cover at level
$(\sigma,\varepsilon)$ obtained by a different route (e.g., a direct
analytic bound as in Theorem~\ref{thm:frr_construction}, or a better
encoder-critic pair with smaller $\eta_{\mathrm{enc}}$); it only rules out
this particular learned representation as a certification vehicle at this
tolerance. In practice, this gives a concrete diagnostic: if a learned
encoder's measured approximation error does not fall below $\varepsilon/2$,
retraining or enlarging the critic --- not shrinking the latent cells ---
is the lever that can restore certifiability.
\end{remark}

Consequently, autoencoder-style belief representations reduce usable
policy capacity by replacing the covering dimension of the original
belief space with the latent dimension $d$, but only after paying the
irreducible encoder error $2\eta_{\mathrm{enc}}$, and only when
$\eta_{\mathrm{enc}}<\varepsilon/2$ strictly. A learned encoder is
therefore useful for FRR exactly when its approximation floor is below
the decision tolerance \emph{and} its latent covering dimension is lower
than that of the belief space --- both conditions are necessary, and
Corollary~\ref{cor:encoder_capacity_comparison} and
Proposition~\ref{prop:encoder_floor} make each one precise.

\subsection{Numerical certification of reliability cells}
\label{subsec:numerical_certification}

The definitions above are constructive: a candidate cover can be certified
by checking whether each cell has sufficiently small decision diameter. In
numerical studies, this can be done either through an analytic Lipschitz
bound, as in Theorem~\ref{thm:frr_construction}, or by direct estimation of
within-cell \(Q^*\)-variation. Given candidate cells \(C_i\), the direct
empirical test is
\[
  \widehat\Delta_Q(C_i)
  =
  \max_{b,b'\in \widehat C_i}
  \max_{u\in\cU}
  |Q^*(b,u)-Q^*(b',u)|,
\]
where \(\widehat C_i\subset C_i\) is a finite set of sampled beliefs. A cell
is empirically certified at tolerance \(\varepsilon\) if
\[
  \widehat\Delta_Q(C_i)\le \varepsilon
\]
with an appropriate sampling margin or confidence interval.

When \(Q^*\) is not available exactly, the same test should be applied with
an explicit approximation margin. For example, suppose a computed critic
\(\widehat Q\) satisfies, on the sampled reachable region,
\[
  \sup_{b,u}|Q^*(b,u)-\widehat Q(b,u)|\le \eta_Q .
\]
Define
\[
  \widehat\Delta_{\widehat Q}(C_i)
  =
  \max_{b,b'\in\widehat C_i}
  \max_{u\in\cU}
  |\widehat Q(b,u)-\widehat Q(b',u)| .
\]
Then the true decision diameter is bounded by
\[
  \Delta_Q(C_i)
  \le
  \widehat\Delta_{\widehat Q}(C_i)
  +2\eta_Q
  +\eta_{\mathrm{samp}},
\]
where \(\eta_{\mathrm{samp}}\) denotes the sampling or generalization margin
needed to pass from the finite test set \(\widehat C_i\) to the whole cell
\(C_i\). Thus empirical certification should report both the observed
within-cell critic variation and the margins used to account for critic and
sampling error.

This procedure applies across several belief representations. In a finite
POMDP, \(Q^*\) can be computed by value iteration or an \(\alpha\)-vector
solver on the belief simplex, and candidate cells can be tested directly by
evaluating \(Q^*\)-variation over sampled beliefs. In linear-Gaussian or
locally linearized models, candidate cells may be chosen as ellipsoids
induced by the sensor-dependent predictive observation pseudometric. In
particle-filter or learned belief representations, both the observation
distinguishability \(d_{\mathrm{obs},\sigma}\) and the belief-kernel
modulus \(\beta_\sigma\) can be estimated from paired rollouts on the
sampled reachable set.

The resulting certified cell count gives an empirical estimate of usable
policy capacity:
\[
  \widehat{\cH}_{\mathrm{FRR}}(\sigma,\varepsilon)
  =
  \log \widehat N_{\mathrm{FRR}}(\sigma,\varepsilon),
\]
where \(\widehat N_{\mathrm{FRR}}\) is the number of certified reliability
cells in the sampled reachable region. Nonlinear or sample-based
experiments should therefore not be presented as evidence of a global
Bayesian contraction. They should be presented as empirical certification
of reachable-set FRR cells: the tested beliefs lie in cells whose observed
\(Q^*\)-variation is below the target decision tolerance.

%============================================================================
\section{Numerical Examples}
\label{sec:numerical}
%============================================================================

This section illustrates how FRR certification behaves when the physical
noise floor affects both sensing and action execution.  The simulation code
used to generate the numerical results is available in the accompanying
GitHub repository~\cite{yoon2026frrgithub}.  The first example is an exact
finite-POMDP computation on a one-dimensional belief simplex.  It is used to
verify the policy-sufficiency mechanism and to separate the effect of sensor
degradation from the effect of action-execution uncertainty.  The next two
examples use particle-filter beliefs for a planar UGV and a two-link arm.
These continuous-state examples are not complete analytical FRR covers; they
are sampled certification diagnostics based on a computable surrogate critic.
The final example is a closed-form modulus illustration showing why monotone
entropy reduction with sensor noise requires additional structure.

%----------------------------------------------------------------------------
\subsection{Exact finite-POMDP certification with sensing and action uncertainty}
\label{subsec:finite_pomdp_numerics}
%----------------------------------------------------------------------------

We first consider a two-state, two-action, two-observation POMDP.  The belief
is represented by
\[
  p=\Pr(s=1), \qquad p\in[0,1].
\]
The reward favors matching the action to the hidden state:
\[
  r(p,a=1)=p, \qquad r(p,a=0)=1-p.
\]
The binary observation channel has correctness probability
\[
  p_{\rm corr}(\sigma_y)
  =
  1-\frac{1}{2}\frac{\sigma_y^2}{\sigma_y^2+\sigma_0^2},
\]
so increasing the sensor noise floor \(\sigma_y\) drives the observation model
toward an uninformative binary sensor.  Action-execution uncertainty is modeled
by an execution channel
\[
  G_{\alpha_u}(\tilde a\mid a)
  =
  \begin{cases}
  1-\alpha_u, & \tilde a=a,\\
  \alpha_u, & \tilde a\neq a,
  \end{cases}
  \qquad 0\le \alpha_u\le \frac12,
\]
where \(a\) is the commanded action and \(\tilde a\) is the executed action.
The Bellman equation is solved on a dense belief grid using the effective
reward and belief-transition kernel obtained by averaging over
\(G_{\alpha_u}\).  For each pair \((\sigma_y,\alpha_u)\), candidate belief
intervals are certified by the sampled decision-diameter test
\[
  \widehat\Delta_Q(C_i)
  =
  \max_{p,p'\in\widehat C_i}
  \max_{a\in\cA}
  \left|
  Q^*_{\sigma_y,\alpha_u}(p,a)
  -
  Q^*_{\sigma_y,\alpha_u}(p',a)
  \right|.
\]
An interval is certified if \(\widehat\Delta_Q(C_i)\le\varepsilon\).  The
representative-belief policy chooses the greedy action at the representative
belief of each certified interval and is evaluated against the optimal value
computed on the same grid.

\begin{figure*}[t]
  \centering
  \includegraphics[width=0.98\textwidth]{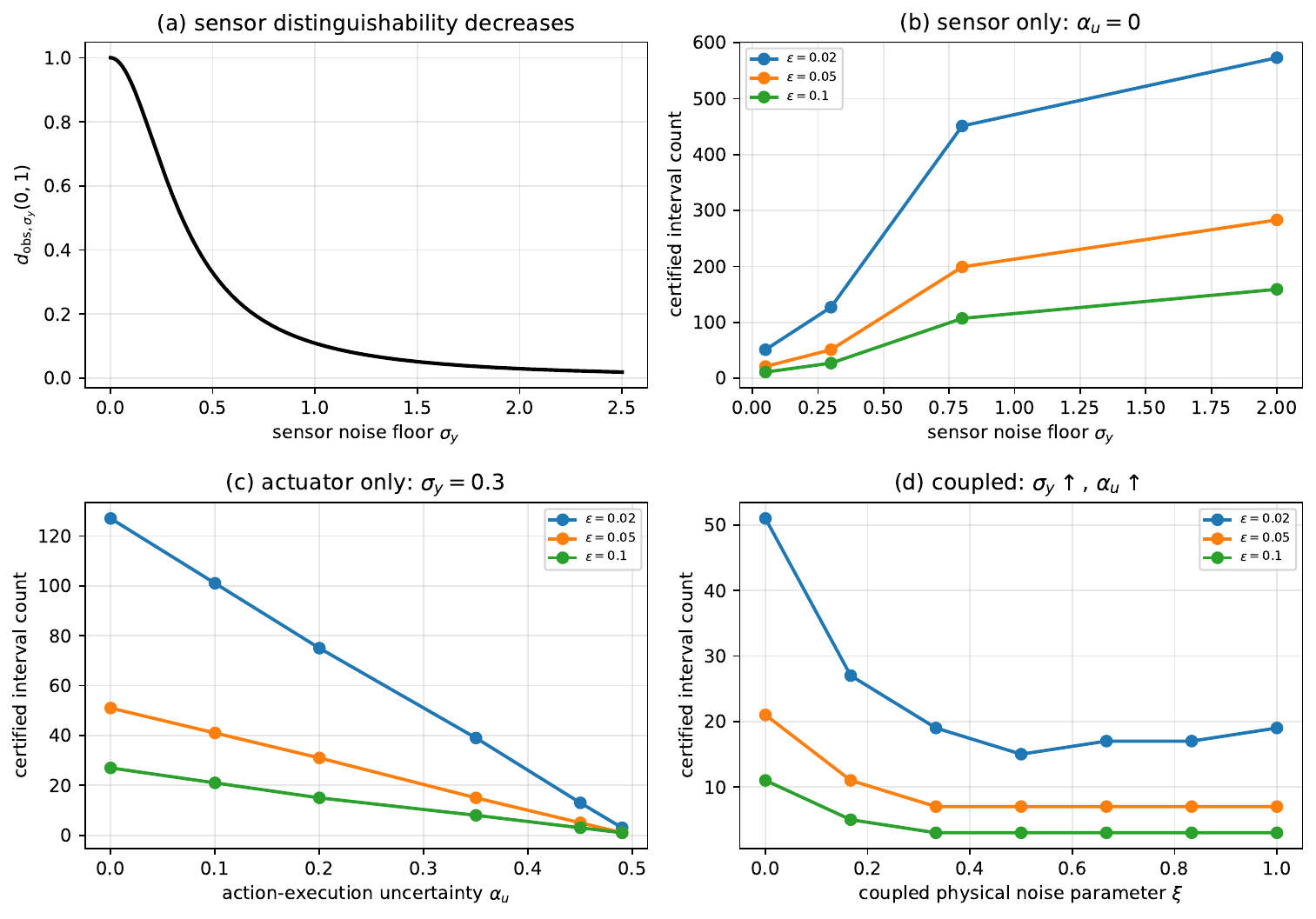}
  \caption{
  Exact finite-POMDP FRR certification with sensing and action-execution
  uncertainty.  Sensor noise alone reduces predictive observation
  distinguishability but does not guarantee a smaller certified cover.  With
  \(\alpha_u=0\), the number of certified intervals for \(\varepsilon=0.05\)
  increases from \(21\) at \(\sigma_y=0.05\) to \(283\) at
  \(\sigma_y=2.0\).  In contrast, at fixed \(\sigma_y=0.3\), increasing
  action-execution uncertainty collapses commanded-action distinctions: the
  maximum action gap decreases from \(1.0\) at \(\alpha_u=0\) to \(0.02\) at
  \(\alpha_u=0.49\), and the \(\varepsilon=0.05\) interval count decreases
  from \(51\) to \(1\).  Along the coupled path
  \((\sigma_y,\alpha_u)=(\sigma_y(\xi),\alpha_u(\xi))\), the interval count
  also drops sharply.  All representative-policy losses remain below the
  policy-sufficiency bound \(2\varepsilon/(1-\gamma)\).
  }
  \label{fig:finite_pomdp_sensing_actuation}
\end{figure*}

Figure~\ref{fig:finite_pomdp_sensing_actuation} shows the main mechanism.
The sensor-only sweep confirms the cautionary point: decreasing observation
distinguishability does not by itself imply a smaller FRR cover.  For example,
when \(\varepsilon=0.05\), the certified interval count increases from
\(21\) to \(283\) as \(\sigma_y\) increases from \(0.05\) to \(2.0\).  The
representative-policy loss remains small, increasing from \(0.0350\) to
\(0.0832\), and remains far below the theorem bound \(2.0\).

The action-execution sweep gives the complementary behavior.  At fixed
\(\sigma_y=0.3\), increasing \(\alpha_u\) makes the two commanded actions
nearly indistinguishable in execution.  The maximum action gap
\(\max_p|Q^*(p,1)-Q^*(p,0)|\) decreases from \(1.0\) at \(\alpha_u=0\) to
\(0.02\) at \(\alpha_u=0.49\).  Correspondingly, the number of certified
intervals for \(\varepsilon=0.05\) decreases from \(51\) to \(1\).  This is
the expected FRR behavior when the physical noise floor limits not only what
can be sensed, but also what can be reliably executed.

The actuator-only panel in Figure~\ref{fig:finite_pomdp_sensing_actuation}
therefore provides the cleanest numerical evidence that action-execution
uncertainty can reduce the number of decision-relevant belief intervals.

The coupled sweep uses a scalar physical degradation parameter \(\xi\), with
both \(\sigma_y\) and \(\alpha_u\) increasing with \(\xi\).  Along this path,
the maximum action gap decreases from \(1.0\) at \(\xi=0\) to approximately
\(0.02\) for large \(\xi\).  For \(\varepsilon=0.05\), the certified interval
count decreases from \(21\) at \(\xi=0\) to \(7\) at \(\xi=1\).  Thus the
finite-POMDP result supports the central interpretation: sensor noise limits
what can be inferred, while action-execution noise limits which commanded
actions can be reliably distinguished; FRR cell count is governed by their
combined effect on \(Q^*\), not by observation distinguishability alone.

%----------------------------------------------------------------------------
\subsection{Particle-filter certification diagnostics}
\label{subsec:particle_filter_numerics}
%----------------------------------------------------------------------------

We next consider two continuous-state examples.  These examples use sampled
particle-filter beliefs and a surrogate critic, so they should be interpreted
as empirical certification diagnostics rather than exact FRR covers.  A cell
that passes the sampled test is certified only for the sampled reachable set
and the critic used in the computation.

For the UGV, the state is \(x\in[0,10]^2\).  The commanded-control dynamics
are
\[
  x_{t+1}=F(x_t,\tilde u_t)+w_t,
  \qquad
  F(x,u)=x+\Delta t\,u,
\]
where \(u\) is the commanded velocity and \(\tilde u\) is the executed
velocity.  Action-execution uncertainty is modeled by a stochastic execution
channel over the discrete action library: as \(\alpha_u\) increases,
\(\tilde u\) is increasingly replaced by a randomly sampled library action.
A fully observed deterministic MDP is first solved on a grid with goal-seeking
reward, producing \(V^*_{\rm MDP}\).  The surrogate critic used for
certification is
\[
\begin{aligned}
  &\widehat Q_{\sigma_y,\alpha_u}(b,u) \\
  &\quad=
  \E_{x\sim b,\tilde u\sim G_{\alpha_u}(\cdot\mid u)}
  \left[
  r\bigl(F(x,\tilde u)\bigr)
  +
  \gamma V^*_{\rm MDP}\bigl(F(x,\tilde u)\bigr)
  \right].  
\end{aligned}
\]
Candidate cells are generated over sampled particle-filter belief means and
are tested by
\[
  \widehat\Delta_{\widehat Q}(C_i)
  =
  \max_{b,b'\in\widehat C_i}
  \max_{u\in\widehat\cU}
  \left|
  \widehat Q_{\sigma_y,\alpha_u}(b,u)
  -
  \widehat Q_{\sigma_y,\alpha_u}(b',u)
  \right|.
\]
A candidate cell is counted as certified if
\(\widehat\Delta_{\widehat Q}(C_i)\le\varepsilon\).

\begin{figure*}[t]
  \centering
  \includegraphics[width=0.9\textwidth]{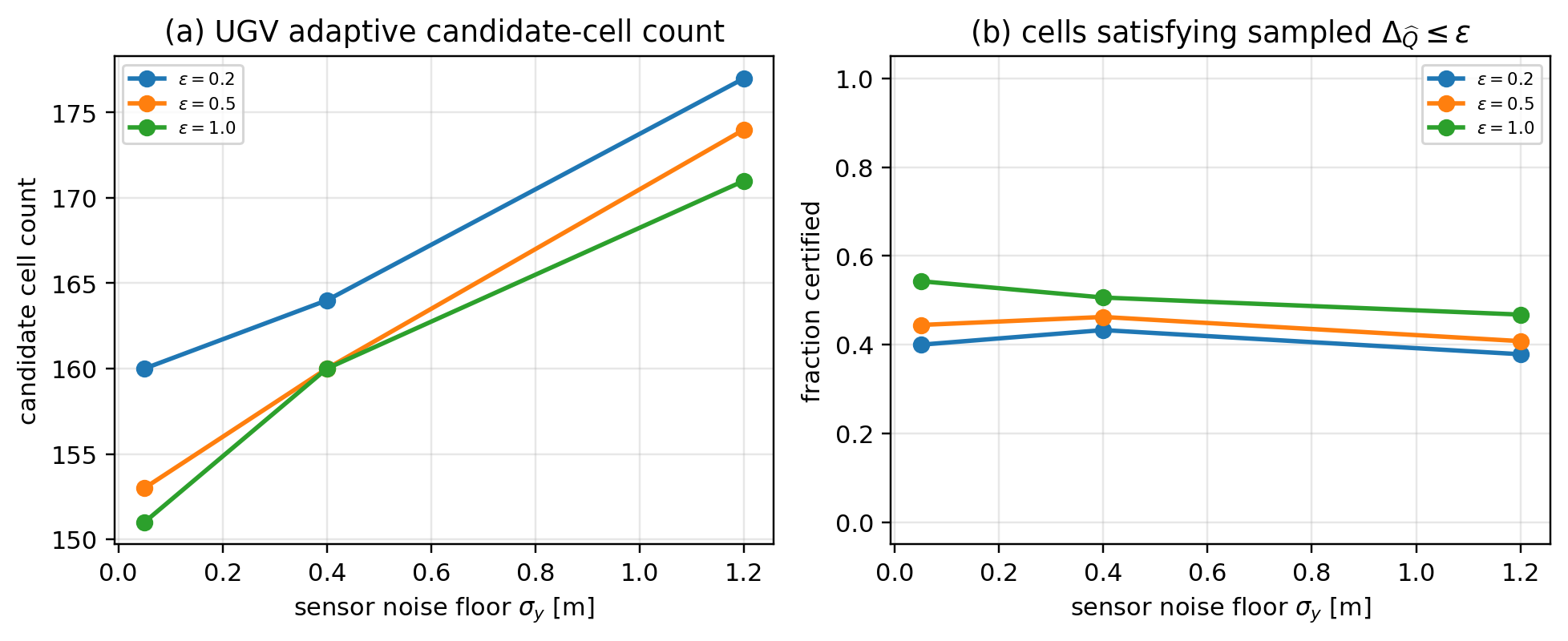}
  \caption{
  UGV sensor-only sampled certification diagnostics with \(\alpha_u=0\).
  The number of adaptive candidate cells increases with \(\sigma_y\), while
  the fraction certified is not monotone.  This mirrors the finite-POMDP
  caution: sensor degradation alone should not be interpreted as a guaranteed
  reduction in FRR cell count.
  }
  \label{fig:ugv_sensor_only}
\end{figure*}

The action-only UGV sweep is therefore summarized in text rather than shown as a separate figure: at fixed \(\sigma_y=0.4\), the mean commanded-action gap decreases from approximately \({2.779}\) at \(\alpha_u=0\) to zero at \(\alpha_u=1\), while the certified fraction changes only mildly.

\begin{figure*}[t]
  \centering
  \includegraphics[width=0.98\textwidth]{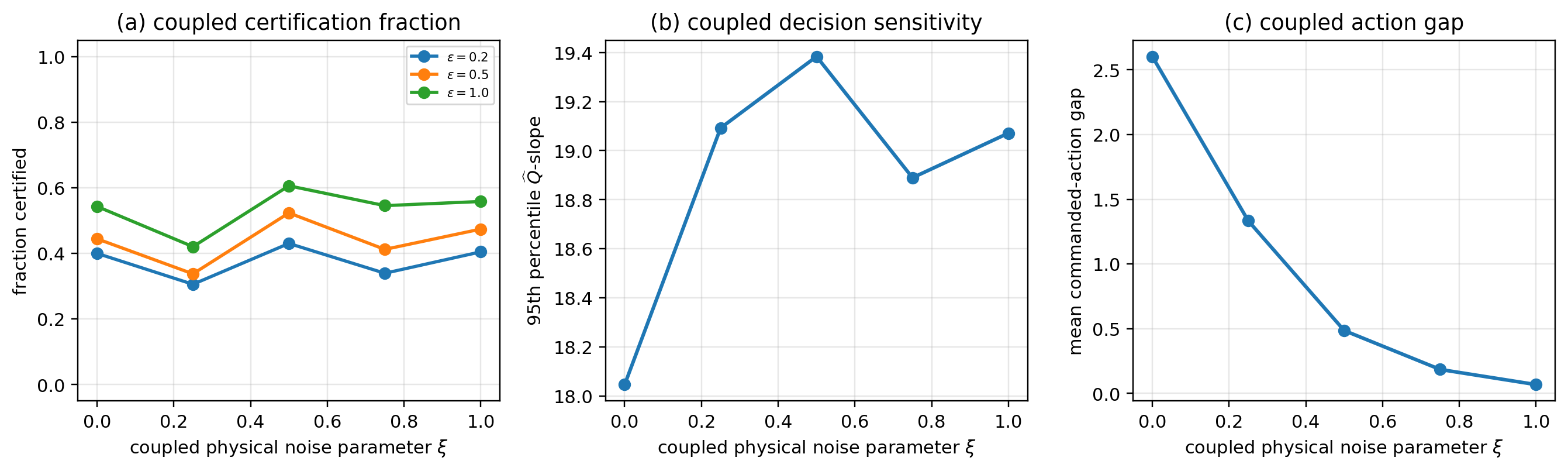}
  \caption{
  UGV coupled sensing/action diagnostic.  As the coupled physical noise
  parameter \(\xi\) increases, the mean commanded-action gap decreases from
  approximately \(2.603\) to \(0.068\).  The maximum unresolved within-cell
  variation also decreases substantially, from about \(16.438\) at \(\xi=0\)
  to about \(8.335\) at \(\xi=1\).  The certified fraction is not monotone,
  which reflects the fact that candidate-cell certification depends on both
  the sampled reachable set and the spatial variation of the surrogate critic.
  }
  \label{fig:ugv_coupled_noise}
\end{figure*}

Figures~\ref{fig:ugv_sensor_only} and~\ref{fig:ugv_coupled_noise} show that the
UGV experiment supports the qualitative distinction between sensing and action
uncertainty.  Sensor noise changes the particle-filter reachable set and does
not produce monotone certification behavior.  Action-execution uncertainty,
however, directly collapses commanded-action distinctions.  The UGV action gap
falls to nearly zero as \(\alpha_u\) approaches one, but the certified fraction
changes modestly because the sampled decision diameter is still dominated by
spatial value variation in some cells.  Therefore, the UGV experiment should be
reported as a sampled diagnostic, not as a proof of monotone FRR entropy.

\begin{figure}[t]
  \centering
  \includegraphics[width=0.82\linewidth]{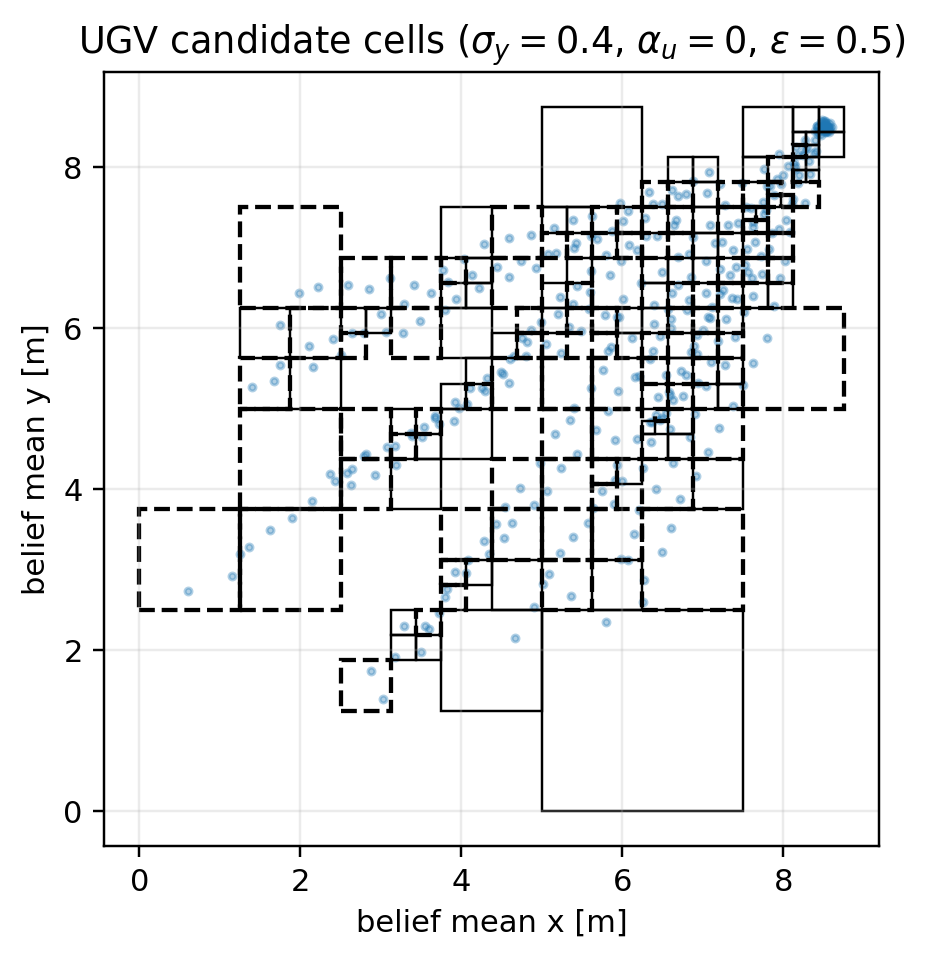}
  \caption{
  UGV candidate cells for \(\sigma_y=0.4\), \(\alpha_u=0\), and
  \(\varepsilon=0.5\).  Blue points are sampled belief means.  Solid boxes
  denote sampled certified cells; dashed boxes denote unresolved cells that
  require further refinement, denser sampling, or improved critic accuracy.
  }
  \label{fig:ugv_candidate_cells}
\end{figure}

The same procedure is applied to a two-link arm in joint space
\(q=(q_1,q_2)\).  The fully observed surrogate MDP is solved in joint space,
with reward determined by end-effector distance to a goal.  The particle
filter represents uncertain joint beliefs, and the surrogate critic averages
the one-step return and interpolated MDP value over both particles and the
action-execution channel.

\begin{figure*}[t]
  \centering
  \includegraphics[width=0.9\textwidth]{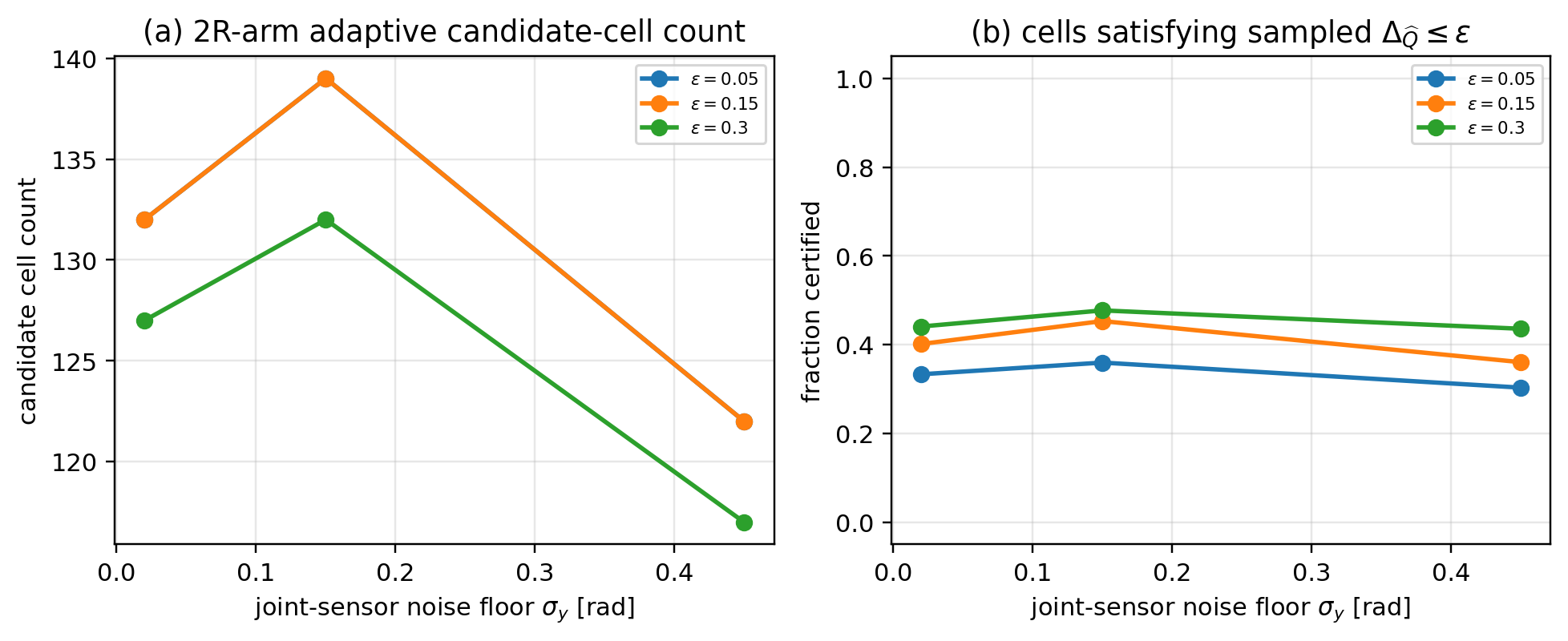}
  \caption{
  Two-link arm sensor-only sampled certification diagnostics with
  \(\alpha_u=0\).  The candidate-cell count and certified fraction vary
  nonmonotonically with joint-sensor noise.  This is consistent with the
  paper's caution that observation degradation alone does not determine FRR
  complexity.
  }
  \label{fig:arm_sensor_only}
\end{figure*}

\begin{figure*}[t]
  \centering
  \includegraphics[width=0.98\textwidth]{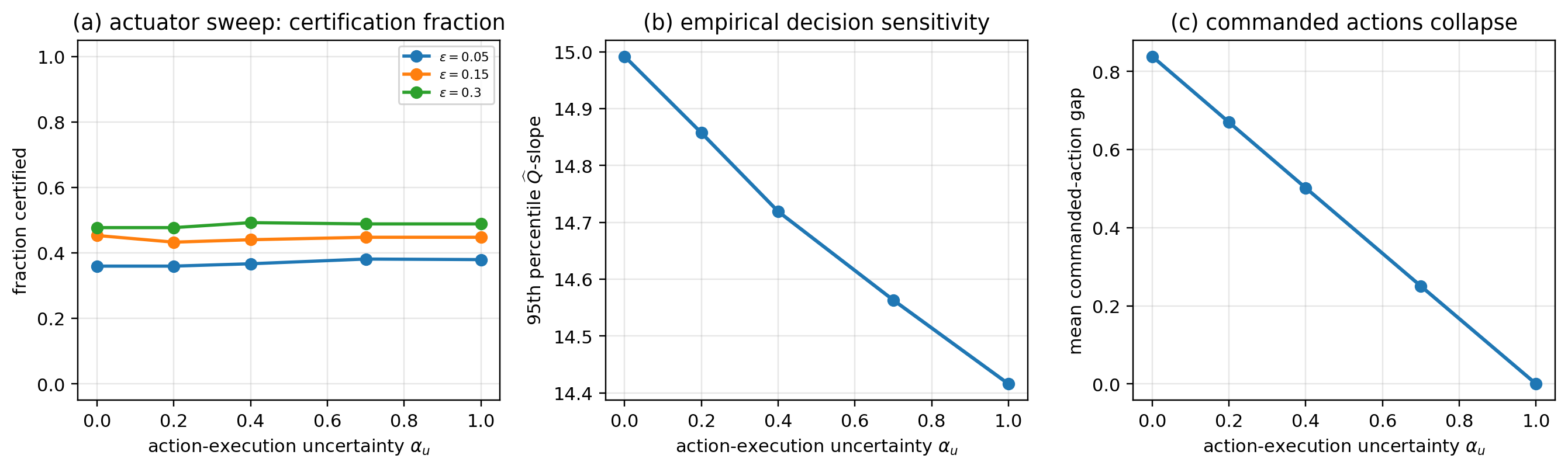}
  \caption{
  Two-link arm action-execution sweep at fixed \(\sigma_y=0.15\).  Increasing
  \(\alpha_u\) reduces the mean commanded-action gap from approximately
  \(0.839\) to zero and decreases the empirical 95th-percentile
  decision-sensitivity statistic from about \(14.992\) to \(14.416\).  The
  certified fraction increases modestly for the tested tolerances, especially
  at the larger decision tolerance.
  }
  \label{fig:arm_actuator_only}
\end{figure*}

\begin{figure*}[t]
  \centering
  \includegraphics[width=0.98\textwidth]{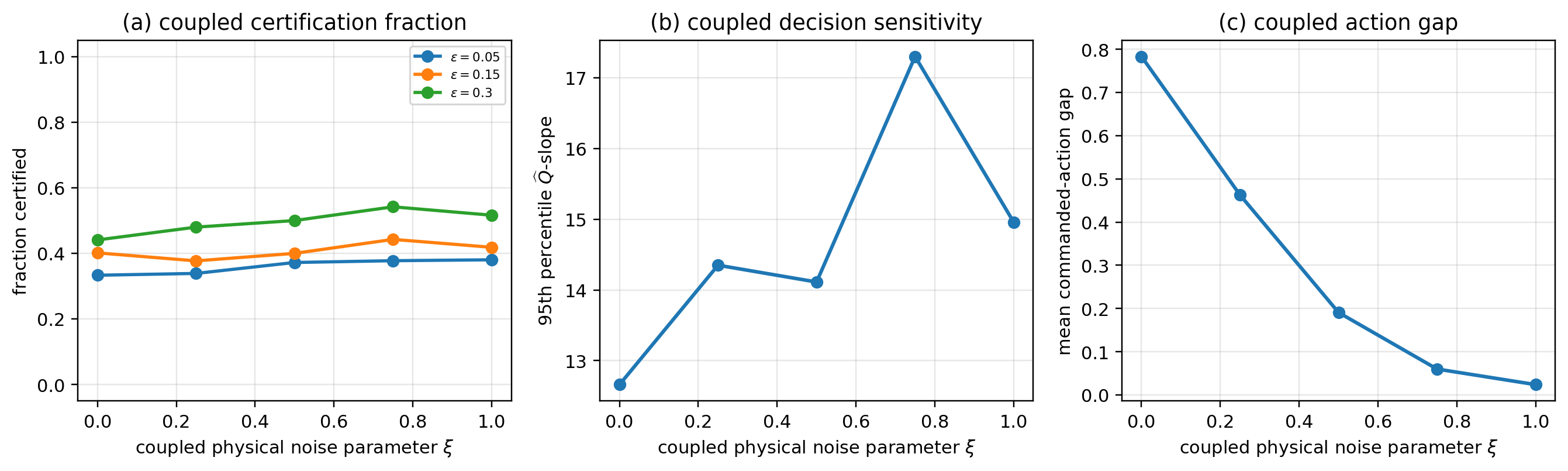}
  \caption{
  Two-link arm coupled sensing/action diagnostic.  The mean commanded-action
  gap decreases from approximately \(0.783\) at \(\xi=0\) to \(0.024\) at
  \(\xi=1\).  The maximum unresolved within-cell variation decreases from
  about \(8.141\) to about \(2.471\), while the certified fraction increases
  for the largest tested tolerance.  The 95th-percentile sensitivity remains
  nonmonotone, reflecting the nonlinear forward-kinematic reward and the
  changing sampled reachable set.
  }
  \label{fig:arm_coupled_noise}
\end{figure*}

\begin{figure}[t]
  \centering
  \includegraphics[width=0.82\linewidth]{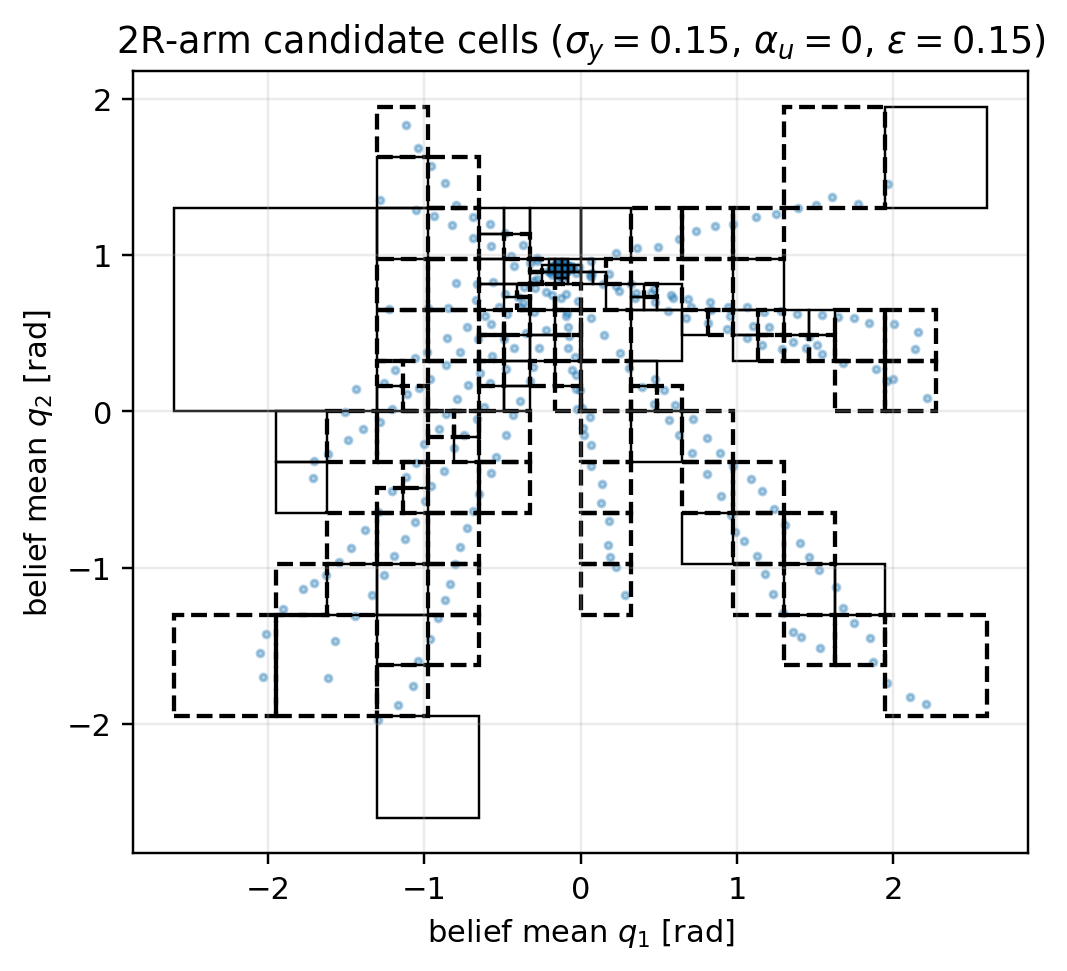}
  \caption{
  Two-link arm candidate cells for \(\sigma_y=0.15\), \(\alpha_u=0\), and
  \(\varepsilon=0.15\).  Blue points are sampled belief means.  Solid boxes
  denote sampled certified cells; dashed boxes denote unresolved cells.
  }
  \label{fig:arm_candidate_cells}
\end{figure}

The arm results show the same qualitative pattern as the UGV results, but
with stronger nonlinear effects.  Action-execution uncertainty reliably
collapses the commanded-action gap.  The sampled certification fraction and
empirical slope statistics are less smooth because the end-effector reward is
nonlinear in joint space and because the reachable particle-belief set changes
with the noise parameters.  These effects reinforce the reason for using the
word ``diagnostic'': the sampled experiments demonstrate how to apply the
FRR decision-diameter test, while the exact finite-POMDP example carries the
strongest certification interpretation.

%----------------------------------------------------------------------------
\subsection{Conditional dependence of entropy bounds on the noise floor}
\label{subsec:conditional_modulus_numerics}
%----------------------------------------------------------------------------

The final example illustrates why monotone decrease of
\(N_{\mathrm{FRR}}(\sigma,\varepsilon)\) with respect to sensor noise should
not be claimed without additional structure.  Consider the illustrative
reachable-set modulus
\[
  \beta_{\mathrm{model}}(\sigma)
  =
  L_f\frac{L_h\sigma_0^2+\sigma^2}{\sigma^2+\sigma_0^2},
\]
with induced value-sensitivity bound
\[
  L_V(\sigma)
  =
  \frac{L_r}{1-\gamma\beta_{\mathrm{model}}(\sigma)}.
\]
This model is not a theorem; it is a controlled illustration of how a
noise-dependent belief-kernel modulus affects the metric-entropy upper bound.

\begin{figure*}[t]
  \centering
  \includegraphics[width=0.9\textwidth]{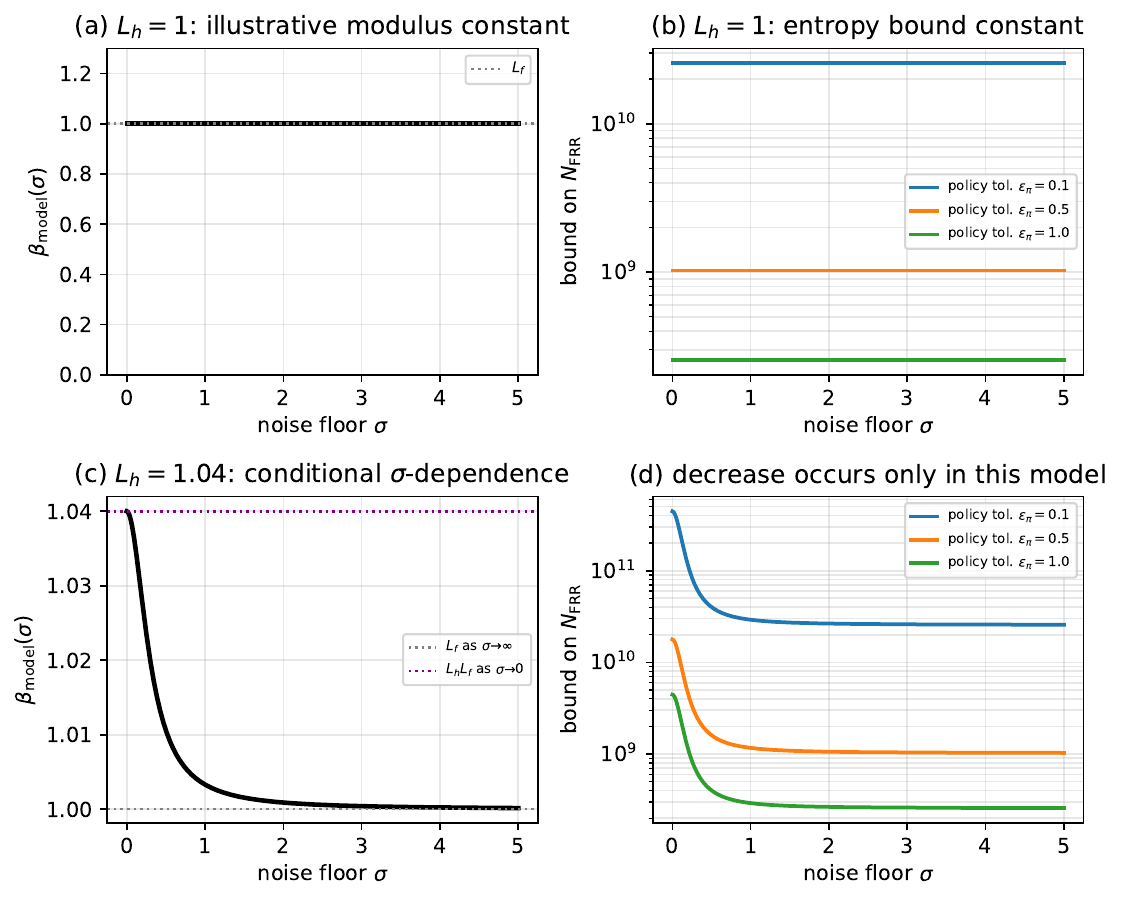}
  \caption{
  Conditional noise-dependence illustration.  When \(L_h=1\), the illustrative
  modulus is constant and the entropy bound is constant.  When \(L_h=1.04\),
  \(\beta_{\mathrm{model}}(\sigma)\) decreases from approximately
  \(1.04\) at small \(\sigma\) toward \(1.0\) at large \(\sigma\), and the
  corresponding FRR cell-count bound decreases.  This example is only a
  model-dependent illustration; it is not a universal monotonicity theorem.
  }
  \label{fig:conditional_noise_dependence}
\end{figure*}

For \(L_h=1\), the model gives
\(\beta_{\mathrm{model}}(\sigma)=1\) and \(L_V(\sigma)=20.0\) at all tested
noise levels.  For \(L_h=1.04\), the value-sensitivity bound decreases from
\(L_V(0.05)=76.763\) to \(L_V(3.0)=20.152\).  This example supports the same
message as the experiments above: observation distinguishability always
decreases under sensor degradation, but reliability entropy decreases only
when the induced decision sensitivity also decreases.

Together, the numerical examples support three conclusions.  First, the exact
finite-POMDP computation verifies the policy-sufficiency mechanism and shows
that sensor noise alone can increase the number of certified cells.  Second,
action-execution uncertainty can collapse commanded-action distinctions and
reduce the number of reliability cells needed, most clearly in the exact
finite-POMDP experiment.  Third, in continuous-state particle-filter examples,
actuation uncertainty visibly reduces commanded-action gaps, but sampled
certification fractions remain affected by reachable-set geometry, critic
approximation, and nonlinear value variation.  Thus, the practical FRR message
is not a universal monotonicity law, but a certification principle: reliability
cells should be calibrated by the combined effect of sensing and action
uncertainty on decision diameter.

%============================================================================
\section{Complementary Noise-Floor Limits}
\label{sec:limits}
%============================================================================

The preceding sections bound the loss caused by replacing the full reachable
belief space with finitely many reliability cells. This is an approximation
bound relative to the optimal policy for the same noisy POMDP model. The
present section addresses a different issue: even the optimal policy for
the noisy POMDP may be separated from an ideal noiseless oracle by a
nonzero performance gap. The following lower-bound examples therefore
clarify what FRR can and cannot do. FRR controls representation-induced
loss, but it cannot remove the irreducible loss imposed by sensing, process, or actuation noise.

\begin{assumption}[Decision-critical pair]
\label{ass:critical_pair}
There exist \(x_1,x_2\in\cX\) and two actions \(u_1\neq u_2\) such that
\[
  r(x_i,u_i)-r(x_i,u_j)\ge \Delta r>0,
  \qquad i\neq j,\quad i,j\in\{1,2\}.
\]
Let
\[
  \Delta h:=\norm{h(x_1)-h(x_2)} .
\]
\end{assumption}

Assumption~\ref{ass:critical_pair} identifies two states that require
different actions: action \(u_i\) is better at state \(x_i\) by at least
\(\Delta r\). If the observations generated by \(x_1\) and \(x_2\) are
nearly indistinguishable under the sensor noise floor, then no decision
rule can reliably choose the correct action in both states.

\begin{theorem}[Sensor information floor]
\label{thm:sensor_floor}
Under the Gaussian observation model and
Assumption~\ref{ass:critical_pair}, consider the single-observation
decision problem
\[
  y=h(x_i)+v,\qquad v\sim\cN(0,\sigma^2 I_p),
\]
where the controller must choose between \(u_1\) and \(u_2\). For any
decision rule \(\delta:\cY\to\{u_1,u_2\}\),
\[
  \max_{i\in\{1,2\}}
  \left[
  r(x_i,u_i)-
  \E r(x_i,\delta(y))
  \right]
  \ge
  \Delta r
  \left(
  \frac12-\frac{\Delta h}{2\sigma\sqrt{2\pi}}
  \right),
\]
whenever the right-hand side is positive.
\end{theorem}

\begin{proof}
Let
\[
  P_i=\cN(h(x_i),\sigma^2 I_p),
  \qquad i=1,2,
\]
be the observation law under state \(x_i\). Any decision rule
\(\delta:\cY\to\{u_1,u_2\}\) induces a binary test between \(P_1\) and
\(P_2\). The worst-case probability of choosing the wrong action is bounded
below by Le Cam's two-point inequality:
\[
  \max_{i\in\{1,2\}}
  \Pr_{y\sim P_i}\{\delta(y)\neq u_i\}
  \ge
  \frac12\bigl(1-\TV(P_1,P_2)\bigr).
\]
For Gaussian distributions with common covariance \(\sigma^2 I_p\),
\[
  \TV(P_1,P_2)
  \le
  \frac{\norm{h(x_1)-h(x_2)}}{\sigma\sqrt{2\pi}}
  =
  \frac{\Delta h}{\sigma\sqrt{2\pi}} .
\]
Thus
\[
  \max_{i\in\{1,2\}}
  \Pr_{y\sim P_i}\{\delta(y)\neq u_i\}
  \ge
  \frac12-
  \frac{\Delta h}{2\sigma\sqrt{2\pi}} .
\]
By Assumption~\ref{ass:critical_pair}, choosing the wrong action at state
\(x_i\) incurs reward loss at least \(\Delta r\). Therefore
\[
  \max_{i\in\{1,2\}}
  \left[
  r(x_i,u_i)-\E r(x_i,\delta(y))
  \right]
  \ge
  \Delta r
  \left(
  \frac12-\frac{\Delta h}{2\sigma\sqrt{2\pi}}
  \right),
\]
whenever the right-hand side is positive.
\end{proof}

Theorem~\ref{thm:sensor_floor} shows that sensor noise creates a
performance floor even for a decision rule with unlimited representation
capacity. This lower bound is distinct from the FRR approximation bound.
FRR controls the loss from using a finite reliability-cell representation
relative to \(V^*\); the information floor controls the loss of the noisy
decision problem itself relative to an idealized noiseless decision maker.

\begin{remark}[A process or actuation noise floor]
A parallel lower bound arises from process or actuator noise. Consider the
standalone benchmark system
\[
  x_{t+1}=x_t+u_t+w_t,
  \qquad
  w_t\sim\cN(0,\sigma_w^2 I_n),
\]
with stage cost
\[
  c(x_{t+1})=\norm{x_{t+1}-x_{\mathrm{goal}}}.
\]
Suppose that, in the absence of process noise, the goal is reachable in one
step. Since the control input must be chosen before \(w_t\) is realized,
every causal policy satisfies
\[
  \E c(x_{t+1})
  =
  \E\norm{x_{t+1}-x_{\mathrm{goal}}}
  \ge
  \E\norm{w_t}
  =
  c_n\sigma_w,
\]
where
\[
  c_n
  =
  \sqrt{2}
  \frac{\Gamma((n+1)/2)}{\Gamma(n/2)}.
\]
The same bound applies if \(w_t\) is interpreted as an execution error,
for example \(x_{t+1}=x_t+u_t+\eta_t\), because the commanded input is
chosen before \(\eta_t\) is realized.
Thus the discounted disturbance cost floor is at least
\[
  \sum_{t=0}^{\infty}\gamma^t \E c(x_{t+1})
  \ge
  \frac{c_n\sigma_w}{1-\gamma}.
\]
This benchmark is separate from the compact-state assumptions used in the
FRR construction; its purpose is only to illustrate that process and
actuation noise impose an irreducible performance floor even with perfect
state representation.
\end{remark}

Together, these examples show why the FRR guarantee should not be
interpreted as eliminating the consequences of physical noise. FRR bounds
the additional loss caused by using a finite reliability-cell
representation. The sensing, process, and actuation-noise bounds above show that even the
best policy for the noisy system may remain separated from an ideal
noiseless or disturbance-free oracle by a nonzero performance gap.

%============================================================================
\section{Discussion}
\label{sec:discussion}
%============================================================================

\subsection{Relation to state aggregation and bisimulation}

Classical state aggregation and bisimulation metrics group physical states
using behavioral similarity criteria. FRR differs in both the space on
which it operates and the criterion used to form cells. First, FRR is a
belief-space construction. This is essential for partially observed systems,
where the controller does not observe the state directly and optimal
policies are functions of information states. Second, FRR is tied to the
physical sensing and actuation model. The predictive observation law
determines which belief differences can be distinguished through the sensor,
while action-execution uncertainty determines which commanded actions can be
reliably separated in their effects. The decision diameter then determines
which of these physically distinguishable or executable distinctions matter
for control.

This distinction also explains why FRR is formulated as a cover rather than
as a quotient space. Approximate decision closeness need not define a
transitive equivalence relation, and overlapping cells may be useful in
practice. What matters for policy sufficiency is not that beliefs belong to
disjoint equivalence classes, but that every cell has uniformly bounded
variation in \(Q^*(b,u)\). Thus FRR should be interpreted as a certified
decision-relevant cover of the reachable belief set, not as an exact
behavioral quotient.

\subsection{Why FRR is not merely value approximation}

Approximating \(V^*\) alone does not guarantee action stability. Two beliefs
may have nearly identical optimal values while requiring different optimal
actions. For example, two obstacle configurations may be equally difficult
in value but require steering in opposite directions. A representation that
preserves only value may therefore be insufficient for reliable control.

FRR instead controls the decision diameter, defined as the largest
within-cell variation of \(Q^*(b,u)\) over all beliefs in the cell and all
admissible actions. This criterion bounds the variation of every action
value inside a cell, not only the optimal value. As a result, an action
selected at a representative belief remains near-optimal throughout the
cell. This is the key reason Theorem~\ref{thm:policy_sufficiency} is
independent of how the cover is obtained: once the decision-diameter
condition holds, the cell-constant policy guarantee follows directly.

This point is especially important when actuation uncertainty is present.
If commanded actions become difficult to distinguish in execution, then
some action-value differences may collapse, allowing coarser reliability
cells. Conversely, if degraded sensing makes belief uncertainty more
decision-critical, finer cells may be required. FRR captures both effects
through \(Q^*\)-variation rather than through sensor distinguishability or
value approximation alone.

\subsection{Relation to learned encoders and policy capacity}

FRR also provides a way to interpret learned belief representations. An
encoder \(E_\psi:\cB_{\mathrm{reach}}\to Z\) can reduce the effective
covering dimension by mapping beliefs into a lower-dimensional latent
space. However, this reduction is useful for decision-making only if the
latent representation preserves action-value variation. As discussed in
Section~\ref{sec:entropy}, a latent cover introduces two distinct sources
of decision-diameter error: an encoder-critic approximation floor and a
latent covering term.

The latent covering term can be reduced by refining the latent cells. The
encoder-critic approximation floor cannot be removed by making the cover
finer unless the learned representation itself improves. Thus,
autoencoder-style compression is compatible with FRR, but only when the
encoder error remains below the decision tolerance. In this sense, FRR gives
a certification criterion for learned belief compression: latent dimensions
are useful only to the extent that they preserve the action-value
distinctions that remain reliable under the physical sensing and actuation
channels.

This perspective also clarifies the meaning of reliability entropy. The
quantity \(\cH_{\mathrm{FRR}}(\sigma,\varepsilon)\) is best interpreted as
a rate-distortion-like measure of certified cell-identity complexity. It
counts how many decision-relevant belief messages must be preserved at
tolerance \(\varepsilon\). It should not be confused with a universal
VC-dimension bound for arbitrary policy classes. For a fixed \(N\)-cell
cover, \(\log N\) measures the code length of the cell identity, while
\(N\) bounds the number of certified regions on which a cell-constant
policy can assign independent decisions.

\subsection{What remains open}

The present paper gives a nonlinear FRR theory conditional on a
reachable-set belief-kernel Lipschitz modulus. Several questions remain
open. First, analytic bounds on \(\beta_\sigma\) should be derived for
specific nonlinear filtering classes under observability, process-noise,
actuation-noise, and mixing assumptions. Second, the relationship between
predictive observation distinguishability, action-execution uncertainty,
belief-kernel smoothness, and decision diameter should be sharpened in
regimes where physical resolution is the dominant source of uncertainty.
Third, the sensing, process, and actuation noise floors in
Section~\ref{sec:limits} should be combined into additive or order-optimal
lower bounds when the corresponding obstructions cannot be jointly
mitigated. Fourth, the encoder-based extension should be developed into a
full theory of learned FRR representations, including online belief-to-cell
assignment, statistical certification of encoder error, and finite-sample
confidence bounds for empirical decision-diameter tests.

%============================================================================
\section{Conclusion}
\label{sec:conclusion}
%============================================================================

We introduced Finite Reliability Representations, a framework for choosing
belief-space resolution according to physical sensing and actuation noise
floors together with decision-relevant action-value variation. The central
object is a reliability-cell cover of the reachable belief space, where
each cell has bounded decision diameter. This criterion directly controls
action-value variation and leads to a policy sufficiency theorem: any
cell-constant policy that acts greedily at representative beliefs is within
\(2\varepsilon/(1-\gamma)\) of the optimal policy for the same noisy POMDP
model.

A key methodological point is that noisy Bayesian updates should not be
treated as globally contractive on arbitrary beliefs. The fixed-observation
Bayesian map can fail to contract belief distance. The appropriate
dynamic-programming object is instead the controlled belief-transition
kernel on the reachable belief set. Under a reachable belief-kernel
Lipschitz condition and the discounted sensitivity condition
\(\gamma\beta_\sigma<1\), the optimal value and action-value functions are
Lipschitz, yielding constructive FRR covers with certified decision
diameter.

The framework separates effects that are often conflated. Predictive
observation distinguishability describes what the sensor can resolve.
Action-execution uncertainty describes which command distinctions can be
reliably realized. Belief-kernel smoothness describes how belief uncertainty
propagates through the planning model. Decision diameter describes which
distinctions matter for control. This separation avoids universal
monotonicity claims about sensor noise: increasing sensor noise always
reduces predictive observation distinguishability, but its effect on
reliability entropy depends on the induced belief dynamics, executable
action resolution, and value sensitivity. In particular, a smaller FRR cover
is expected only when the physical noise floor collapses decision-relevant
distinctions in \(Q^*\), not merely because observations become less
informative.

The resulting theory applies directly to finite POMDPs and provides a
basis for linear-Gaussian, locally linearized, particle-filter, and learned
belief representations. Reliability entropy quantifies the minimum number
of certified decision-relevant belief messages required at tolerance
\(\varepsilon\), while complementary noise-floor limits show that sensing,
process, and actuation noise impose irreducible performance gaps even for
policies with unlimited representation capacity. Together, these results
provide a foundation for noise-calibrated belief-space abstraction and for
future empirical certification of reliable finite representations in
autonomous systems.

\bibliographystyle{IEEEtran}
\bibliography{references}

\end{document}